\documentclass[%
twocolumn]{revtex4-2}

\usepackage{amsmath, amssymb, amsfonts}
\DeclareMathOperator{\Tr}{Tr}
\usepackage{bm}
\usepackage{bbm} 
\usepackage{amsthm}
\theoremstyle{plain}
\newtheorem{lemma}{Lemma}[section]
\newtheorem{claim}[lemma]{Claim}

\theoremstyle{definition}

\theoremstyle{remark}

\usepackage{graphicx}
\usepackage{booktabs}
\usepackage{multirow}
\usepackage{caption}
\usepackage{subcaption}
\usepackage{float}

\usepackage{enumitem}
\usepackage{xcolor}
\usepackage{siunitx}
\usepackage{hyperref}
\hypersetup{
    colorlinks=true,
    linkcolor=blue,
    citecolor=blue,
    urlcolor=blue
}

\usepackage{comment}

\usepackage{xurl}

\begin{document}

\title{On multipoles, their decomposition by time-reversal symmetry,
and the electric toroidal monopole}

\author{Vinzenz A. Müller}
\affiliation{Materials Theory, ETH Z{\"u}rich, 8093 Zurich, Switzerland}
\author{Nora Taufertshöfer}
\affiliation{Materials Theory, ETH Z{\"u}rich, 8093 Zurich, Switzerland}
\author{Nicola A. Spaldin}
\affiliation{Materials Theory, ETH Z{\"u}rich, 8093 Zurich, Switzerland}


\begin{abstract}
The multipole decomposition of the single-site density matrix provides a symmetry-adapted representation of local electronic degrees of freedom. Conventional, so-called fixed-shell, formulations do not span the full local single-particle operator space, as only operators mapping within the same orbital $l$ are resolved. Here we construct a complete orthogonal basis of real Hermitian multipole operators for the local density matrix by extending the existing formulation to inter-shell operators. 
We revisit the multipole decomposition as a decomposition of the operator space by $\mathrm{SO}(3)$ by first decomposing the orbital and spin operator spaces. Then by coupling them we arrive at the spin-$\frac12$ local single-particle operator space, staying consistent with the existing fixed-shell formulations.
We then classify the multipoles by parity and time-reversal symmetry, which allows for unique identification of multipole moments of the density matrix that contribute to expectation values of fully symmetry-resolved observables.
As an application, we analyze the two enantiomers of chiral trigonal tellurium by computing the electric toroidal monopole moment selected by symmetry.
\end{abstract}

\maketitle

\section{Introduction}
\label{sec:introduction}

Multipole representations provide a convenient symmetry-resolved language for local electronic degrees of freedom. 
By expanding local operators in irreducible spherical tensors, charge, magnetic, orbital, and spin-orbital degrees of freedom 
can be classified according to their transformation properties under rotations, spatial inversion, and time-reversal. 
This idea is rooted in angular-momentum theory and the theory of irreducible tensor operators 
\cite{Racah1942,Wigner1959,Edmonds1957,BrinkSatchler1968,varshalovich1988}. 
Although a crystalline environment reduces the continuous rotational symmetry of an isolated atom to a finite point group, 
the spherical multipole basis remains a useful parent classification: the point-group-allowed local degrees of freedom are obtained 
by further decomposing the $\mathrm{SO}(3)$ tensors into irreducible representations of the site symmetry.\\
This approach is particularly useful in first-principles density functional calculations, where the local electronic state is naturally 
represented by a projected single-particle density matrix. 
A multipole expansion of this density matrix converts its matrix elements into symmetry-labeled amplitudes. 
Such decompositions have been used to analyze orbital and spin-orbital polarization channels 
\cite{vanderLaan1995,Bultmark2009,Cricchio2009}, 
and first-principles approaches have been developed for magnetic, toroidal, and magnetoelectric multipoles in both 
insulating and metallic systems 
\cite{Ederer2007,Spaldin2008,Spaldin2013,TholeFechnerSpaldin2016,Thole2018,BhowalSpaldin2021,UrruSpaldin2022,Urru2023,VerbeekUrruSpaldin2023}. 
More recent multipole formulations classify local electronic degrees of freedom into 
electric, magnetic, magnetic-toroidal, and electric-toroidal sectors according to their spatial-inversion and time-reversal parities 
\cite{Winkler, Hayami2018,Hayami2018PRB,Kusunose2020,Yatsushiro2021,Hayami2024PSJ}, many of which are not included in a fixed-shell multipole formulation.\\
In this work we revisit the decomposition by $\mathrm{SO}(3)$ of the local single-particle density matrix
and construct a complete orthogonal basis of real Hermitian multipoles for the local operator space.
We further classify these multipoles by spatial parity and time-reversal symmetry,
thereby organizing the full density-matrix information into symmetry-distinct channels.
Our construction closely follows the standard angular-momentum coupling scheme,
as done in Refs. \cite{vanderLaan1995} and \cite{Bultmark2009} for fixed-shell multipoles to which it reduces for $l'=l$,
and extends it to include missing inter-shell operators yielding a complete set of multipoles spanning the operator space.
This extended classification exposes inter-orbital multipoles that are not present in a fixed-shell analysis. 
For example, in the odd-parity, time-reversal-even sector of this extended classification, the lowest-rank pseudoscalar is the electric toroidal monopole. 
Because it is invariant under proper rotations but changes sign under inversion and mirror operations, 
the electric toroidal monopole has the symmetry of a chiral scalar. 
It has therefore been proposed as a microscopic measure of chirality in quantum systems 
and as a symmetry-adapted descriptor of handed electronic structures \cite{Inda2024,Kusunose_chirality_APL_2024,Oiwa2025PRR,Spaldin2026}.
\\
As an illustrative application, we consider the two enantiomers of trigonal tellurium, a paradigmatic elemental chiral crystal exhibiting chirality-dependent spin textures 
and current-induced magnetoelectric responses \cite{Furukawa2017,Sakano2020,Tsirkin2018,Gatti2020,Calavalle2022,Roy2022, Oiwa2025PRR}. 
We demonstrate how a complete multipole decomposition of the local density matrix of tellurium gives access 
to physically interpretable symmetry components, 
with particular emphasis on the inter-shell contributions that encode the electric toroidal monopole.\\
The paper is organized as follows. 
In Sec.~\ref{sec:basis}, we construct the multipole basis for the local single-particle density matrix. 
We first decompose the orbital and spin operator spaces separately and then couple them to obtain spin-orbital multipoles. 
We then derive the real Hermitian basis appropriate for both intra-shell and inter-shell blocks. 
In Sec.~\ref{sec:decomp_parity_time}, we classify the multipoles according to parity and time-reversal symmetry. 
In Sec.~\ref{sec:casestudy}, we apply the formalism to the two enantiomers of tellurium, 
focusing on the electric toroidal monopole as a local pseudoscalar measure of electronic chirality. 
Finally, we summarize the main results and discuss how the construction can be used to analyze symmetry-resolved density matrices in first-principles calculations.

\section{Multipoles as a Basis}
\label{sec:basis}

\noindent 
We consider the single-particle Hilbert space \(\tilde{\mathcal{H}}\) of an electron localized on a site. 
Its angular degrees of freedom are conveniently described by orbital components,
that arise from the decomposition of the angular dependence of wavefunctions 
into irreducible representations of $\mathrm{SO}(3)$. This angular Hilbert space, naturally represented by the space of square integrable functions on the sphere $\mathrm{S}^2$, i.e.
$L^2(\mathrm{S}^2)$, decomposes as
\begin{equation}
L^2(\mathrm{S}^2) = \bigoplus_{l=0}^\infty \Phi_l,
\end{equation}
where $\Phi_l$ is the $(2l+1)$-dimensional irreducible representation spanned by spherical harmonics 
$Y_l^m(\theta,\phi)$ with $m=-l,\dots,l$ and $\bigoplus$ represents the direct sum. These functions provide a position-space realization ($\langle \theta,\phi \,|\, l,m \rangle = Y_l^m(\theta,\phi)$)
of the orbital states $|l,m\rangle$, which for a fixed angular momentum $l$ constitute a vector space $V_l$.  \\

In the full wavefunction, the angular dependence is accompanied by a radial part $R$, which is determined by the site's central potential (or pseudopotential in many practical calculations).
The spin-less single-particle wavefunctions take the form $R_{nl}(r) Y_l^m(\theta,\phi)$, where $n$ is the principal/radial quantum number.

The spin degrees of freedom span a space $V_s$ of dimension $2s+1$. 
The single-particle spin-orbital Hilbert space can thus be written as
\begin{equation}
\tilde{\mathcal{H}} = \left( \bigoplus_{n,l} V_{n,l} \right) \otimes V_s,
\end{equation}
with basis states $|n,l,m\rangle \otimes |s,m_s\rangle$.\\

In the present context, we are interested in the single-site single-particle density matrix we obtain from density functional theory (DFT) calculations,
by projection  onto a chosen finite-dimensional
local orbital subspace \(\mathcal{H}\). The density matrix is a Hermitian, positive semidefinite operator in the space of endomorphisms on $\mathcal{H}$, \(\mathrm{End}(\mathcal{H})\), whose trace gives the occupation of the site \cite{Blochl1994PAW,KresseJoubert1999PAW}.
\\

The restriction of the full single-particle Hilbert space \(\tilde{\mathcal{H}}\) to a finite dimensional subspace \(\mathcal{H}\) occurs due to the consideration of only a finite set of orbital quantum numbers and the use of only one radial projector per angular-momentum shell. 
With \(\mathcal{H}\) finite-dimensional, \(\mathrm{End}(\mathcal{H})\) is isomorphic to the tensor-product space
\begin{equation}
\mathrm{End}(\mathcal{H}) \cong \mathcal{H}\otimes \mathcal{H}^*,
\end{equation}
with the map \(
|u\rangle \langle v| \mapsto |u\rangle \otimes \langle v|,
\)
for \(|u\rangle,|v\rangle \in \mathcal{H}\), and identifying the dual of \(|v\rangle\) with the bra \(\langle v|\in \mathcal{H}^*\) via the inner product on \(\mathcal{H}\).
Using the canonical associativity and permutation isomorphisms of tensor products, 
the operator space can be rearranged as
\begin{equation}
\mathrm{End}(\mathcal{H}) 
\cong 
 \left( \bigoplus_{n,l,n',l'} 
 \left( V_{n,l} \otimes V_{n',l'}^*\right) \right)
\otimes 
\left( V_s \otimes V_s^* \right),
\label{eq:reshuffled_operator_space}
\end{equation}
where orbital and spin degrees of freedom are separated.\\

$\mathcal{H}\otimes \mathcal{H}^*$ can be understood as a representation of $\mathrm{SO}(3)$, and thus admits a decomposition into irreducible subspaces, as do the orbital operator spaces $V_{n,l} \otimes V_{n',l'}^*$ and the spin operator space $V_s \otimes V_s^*$.
While the radial functions $R_{nl}(r)$ depend on the details of the local potential, 
they transform trivially under rotations. As a result, the decomposition into 
irreducible representations is entirely governed by the angular part. 
Going forward, we hence omit the quantum number $n$ and use $l$ and $l'$ in place of $(n,l)$ and $(n',l')$. The radial dependence would enter the formalism only as an additional scalar factor.\\
The decomposition of the operator space as in eq.~\eqref{eq:reshuffled_operator_space}  provides a structure for separately decomposing the orbital and spin parts  into irreducible subspaces before forming their product according to the addition of angular momenta, respecting the structure of irreducible subspaces, called coupling of orbital and spin operators.

For a given $l$ and $l'$, we obtain the irreducible subspaces $W_i$ of 
$ \left( V_l \otimes V_{l'}^* \right)
\otimes 
\left( V_s \otimes V_s^* \right)$.
Each such subspace is given by a so-called multipole tensor. More precisely,
the irreducible subspace $W_i$ is spanned by the components of the multipole tensor.

However, we are not interested in the complete $\mathbb{C}$-linear space, but only in the $\mathbb{R}$-linear Hermitian subspace. 
The goal will therefore be to find a basis of $W_i$ consisting of Hermitian operators in which the density matrix can be expanded, see appendix \ref{app:multipole_decomp}.
As discussed in subsections \ref{subsec:decomp_orbital_tensor} and \ref{subsec:real}, to construct a Hermitian operator, we are required to consider the direct sum of the orbital spaces $\left( V_l \otimes V_{l'}^* \oplus V_{l'} \otimes V_{l}^*\right)$, noting that
\begin{equation}
\begin{aligned}
    \bigoplus_{l',l}V_l \otimes V_{l'}^*\cong & \left(\bigoplus_{l'< l} \left(V_l \otimes V_{l'}^*\right)\oplus \left(V_{l'} \otimes V_{l}^*\right)\right)\\
    &\oplus \left(\bigoplus_l V_l \otimes V_{l}^*\right).
\end{aligned}
\end{equation}
Hence, we will find the Hermitian operators in $\left( (V_l \otimes V_{l'}^*)\otimes (V_s \otimes V_s^*) \oplus (V_{l'} \otimes V_{l}^*)\otimes (V_s \otimes V_s^*)\right)$.

\subsection{Decomposition of the orbital operator space  by \texorpdfstring{$\mathrm{SO}(3)$}{SO(3)}}
\label{subsec:decomp_orbital_tensor}
Operators $T$ mapping from $V_{l'}$ to $V_l$ inherit a natural action of rotations from the 
transformation properties of states. Under a rotation $R \in \mathrm{SO}(3)$, 
the states transform as
\begin{equation}
|l,m\rangle \;\mapsto\; \sum_{m'} D^{(l)}_{m'm}(R)\,|l,m'\rangle,
\end{equation}
where $D^{(l)}_{m'm}(R)$ denotes the Wigner-D matrix, a unitary square matrix as expressed in the orbital basis \cite[Ch.~3.5]{sakurai} \cite[Ch.~4.1]{Edmonds1957}.
Requiring that the physics is the same in all rotated frames, i.e.
 $\langle \psi_R | T_R | \phi_R \rangle = \langle \psi | T | \phi \rangle$,
enforces the transformation law
\begin{equation}
T \;\mapsto\; D^{(l)}(R)\,T\,D^{(l')}(R)^{-1}.
\end{equation}
for the operator \cite[Ch.~5.2]{Edmonds1957}.
Thus, the space of operators $T : V_{l'} \to V_l$ constitutes a representation 
of $\mathrm{SO}(3)$ itself, and by the Clebsch-Gordan decomposition, 
this representation splits into irreducible components,
\begin{equation}
\mathrm{Hom}(V_{l'},V_l) \cong \bigoplus_{k=|l-l'|}^{l+l'} V_k.
\end{equation}
Accordingly, $\mathrm{Hom}(V_{l'},V_l)$ admits a basis consisting of irreducible spherical rank-$k$ tensor operators \(C^k(l',l)\), or rather their components
\begin{equation}
    C^k_q(l',l) : V_{l'}\to V_l, \qquad q=-k,\ldots,k,
    \label{eq:spherical_tensors}
\end{equation}
where the labels $(l',l)$ indicate that $C^k_q(l',l)$ maps $V_{l'}$ to $V_l$, distinguishing tensors of the same rank $k$ acting between different orbital sectors.
The tensor components transform under rotations as
\begin{equation}
    D^{(l)}(R)\,C^k_q(l',l)\,D^{(l')}(R)^{-1} 
    = \sum_{q'=-k}^{k} D^{(k)}_{q'q}(R)\,C^k_{q'}(l',l),
\end{equation}
i.e. each subspace $V_k\in \mathrm{Hom}(V_{l'},V_l)$, which is spanned by the components of the rank-$k$ spherical tensor operator, is invariant under rotations, and the components transform among themselves in the same way as the spherical harmonics \cite[Eq.~5.2.1]{Edmonds1957}.

The Wigner–Eckart theorem fixes the dependence of the matrix element of an orbital operator on
$m$, $m'$, and $q$, up to a convention-dependent reduced matrix element.

Using the Condon-Shortley phase convention for the orbital states, the Clebsch-Gordan coefficients and Wigner-3j symbols are real and related as \cite[Eq.~3.7.3]{Edmonds1957}
\begin{equation}
\langle l'm';kq|lm\rangle
=
(-1)^{-l'+k-m}\sqrt{2l+1}
\begin{pmatrix}
l' & k & l\\
m' & q & -m
\end{pmatrix}.
\end{equation}

The Wigner-Eckart theorem then relates the matrix elements to the Clebsch-Gordan coefficients as
\begin{equation}
\langle lm |C^k_q(l',l)|l'm'\rangle
=
\langle l'm';kq|lm\rangle
\langle l\Vert C^k(l',l)\Vert l'\rangle_{\rm CG},
\end{equation}
or equivalently with Wigner-3j symbols \cite[Eq.~5.4.1]{Edmonds1957} as
\begin{equation}
\begin{aligned}
&\langle lm |C^k_q(l',l)|l'm'\rangle \\
&= 
 (-1)^{l-m}
\begin{pmatrix}
l & k & l'\\
-m & q & m'
\end{pmatrix}
\langle l\Vert C^k(l',l)\Vert l'\rangle_{3j}.
\label{eq:wigner_eckart}
\end{aligned}
\end{equation}

The two reduced matrix elements $\langle \cdot || \cdot || \cdot \rangle \in \mathbb{C}$ denoted by the subscripts $\mathrm{CG}$ and $3j$ differ only by a convention-dependent phase and normalization factor.
The latter form, based on Wigner-3j symbols, is the one used in the formulation by van der Laan \emph{et al.} \cite[Eq.~A4]{vanderLaan1995} and Bultmark \emph{et al.} \cite[Eq.~21]{Bultmark2009}, which is the fixed-shell formulation that we are extending.

Note the difference in ordering of the elements between the Clebsch-Gordan and Wigner-3j versions of the Wigner-Eckart theorem. This stems from the different physical origins of the expressions. While the matrix elements, as well as the Wigner-3j symbols read $\langle$Final$|$Operator$|$Initial$\rangle$, the Clebsch-Gordan coefficients read $\langle$Initial;Operator$|$Final$\rangle$.

We are still left with a choice of phase and normalization for the reduced
matrix elements, for which several conventions are used in the literature.
A particularly simple choice with regard to purely algebraic manipulations would be
\begin{equation}
\langle l\Vert C^k(l',l)\Vert l'\rangle_{3j} = 1,
\end{equation}
as used in Racah's unit tensor \cite{Racah1942,Duros_2025}.
One could also identify the tensor with a spherical harmonic via $C^k_q(l',l)(\theta, \phi)=\sqrt{4\pi/(2k+1)} Y_k^q(\theta, \phi)$
and compute its components via Gaunt integrals 
\cite[Eq.~4.6.3]{Edmonds1957}:
\begin{equation}
\begin{aligned}
&\langle lm |C^k_q(l',l)|l'm'\rangle\\
&=\sqrt{4\pi/(2k+1)}\int_{S^2} Y_l^m{}^* Y_k^q Y_{l'}^{m'}d\Omega\\
&=\sqrt{4\pi/(2k+1)}(-1)^m\int_{S^2} Y_l^{-m} Y_k^q Y_{l'}^{m'}d\Omega\\
&=(-1)^m\sqrt{2l+1}\sqrt{2l'+1}
\begin{pmatrix}
l & k & l'\\
0 & 0 & 0
\end{pmatrix}
\begin{pmatrix}
l & k & l'\\
-m & q & m'
\end{pmatrix}.
\end{aligned} 
\end{equation}
Hence the reduced matrix element would take the form
\begin{equation}
\langle l\Vert C^k(l',l)\Vert l'\rangle_{3j}
=(-1)^{l+2m}\sqrt{2l+1}\sqrt{2l'+1}
\begin{pmatrix}
l & k & l'\\
0 & 0 & 0
\end{pmatrix}.
\end{equation}
Since the Wigner-3j symbol
\(
\begin{pmatrix}
l & k & l'\\
0 & 0 & 0
\end{pmatrix}
\)
vanishes unless $l+k+l'$ is even, so would all orbital operators.

van der Laan \emph{et al.} \cite[Eq.~A4]{vanderLaan1995} and Bultmark \emph{et al.} \cite[Eq.~21]{Bultmark2009} instead use the normalization
\begin{equation}
\langle l\Vert C^k(l,l)\Vert l\rangle_{3j}
=
n_{lk}^{-1},
\end{equation}
with
\begin{equation}
n_{lk}=\frac{(2l)!}{\sqrt{(2l-k)!(2l+k+1)!}}
\label{eq:n_lk}
\end{equation}
for $l'=l$, which ensures that tensor operators of all ranks $k$ are retained.
This convention yields a complete set of irreducible tensor operators within a fixed $l$-shell and is particularly convenient when expressing the density matrix in terms of multipole moments, as it achieves
\begin{equation}
\langle l l|C^k_0|ll\rangle=1.
\end{equation}
That is, the $q=0$ component is normalized to unity on the stretched
orbital state $m=l$. This follows from
\begin{equation}
n_{lk}=
\begin{pmatrix}
    l & k & l \\
    -l & 0 & l
\end{pmatrix}
\end{equation}
as follows directly from the closed form of the Wigner-3j symbol; see also \cite[Footnote~12]{Edmonds1957}.

We can adapt this reduced matrix element to general $l$ and $l'$:
\begin{equation}
\langle l\Vert C^k(l',l)\Vert l'\rangle_{3j}
= n_{lkl'}^{-1},
\end{equation}
with
\begin{equation}
\begin{aligned}
n_{lkl'} &= \sqrt{\frac{(2l)!(2l')!}{(l+l'-k)!(l+l'+k+1)!}}\\
&=
\begin{pmatrix}
    l & k & l' \\
    -l & l-l' & l'
\end{pmatrix},
\end{aligned}
\label{eq:orbital_normalization}
\end{equation}
resulting in a normalization for the stretched matrix element
\begin{equation}
    \langle l,l|C^k_{l-l'}(l',l)|l',l'\rangle=1.
\end{equation}

With this normalization the orbital spherical tensor operator becomes
\begin{equation}
    C^k_q(l',l) = \sum_{m,m'}n_{lkl'}^{-1}(-1)^{l-m}
\begin{pmatrix}
l & k & l'\\
-m & q & m'
\end{pmatrix} |lm\rangle \langle l'm'|.
\label{eq:orbital_tensor_operator}
\end{equation}
Hermitian conjugation of the orbital tensor components gives
\begin{equation}
    \left( C^k_q(l',l)\right)^\dagger = (-1)^{l-l'-q}C^k_{-q}(l,l'),
    \label{eq:hermconj_orbitaltensor}
\end{equation}
as shown in the appendix \ref{app:hermconj_orbitaltensor}.
The Hermitian adjoint lives in the opposite $(l,l')$ block, since with $C^k_q(l',l) \in \mathrm{Hom}(V_{l'},V_l)$, we have $\left( C^k_q(l',l)\right)^\dagger\in\mathrm{Hom}(V_l, V_{l'})$.

Orthogonality of the spherical tensor components with respect to the Hilbert-Schmidt inner product can be shown based on representation theory using Schur's lemma or by direct calculation as done in the appendix \ref{app:orthogonality_orbital}.

The components of the orbital tensor operator defined by eq. \eqref{eq:orbital_tensor_operator} thus form a complete orthogonal basis set of the operator space $\mathrm{Hom}(V_{l'},V_l)\cong V_l \otimes V_{l'}^*$ over \(\mathbb{C}\), and reduce to the operators as defined in \cite{vanderLaan1995} for $l'=l$.
Note, one could add the factor of $i^{l+l'}$, symmetric in $l'$ and $l$, to remove the $(l',l)$-dependence in the sign for the Hermitian adjoint. But in doing so, the $C^k_0$-component, i.e. $z$-component, of the stretched state would pick up an $l$-dependent sign
\begin{equation}
\langle l l|C^k_0|ll\rangle=i^{-2l}=(-1)^l.
\end{equation}
An anti-symmetric factor like $i^{l-l'}$ would recover the expectation value of 1 again, but not affect the Hermitian conjugation relation.
Hence we choose to keep the normalization constant real and carry on the $(l',l)$-dependence in sign of the Hermitian conjugation relation.

\subsection{Decomposition of the spin operator space by \texorpdfstring{$\mathrm{SO}(3)$}{SO(3)}}
\label{sec:decomp_spin}

The Hermitian operators in the spin space $V_s \otimes V_s^*$, with \(s=\frac12\), are spanned by the identity 
and the Pauli matrices. The spin-operator space decomposes as
\begin{equation}
V_s\otimes V_s^* \simeq \frac12\otimes \frac12 = 0\oplus 1 .
\end{equation}
Hence the four spin operators may be organized into irreducible spherical tensors \(\sigma^p_y\), with \(p=0,1\):

\begin{equation}
    \sigma^p_y = \sum_{m_s,m_s'}n_p^{-1}(-1)^{s-m_s}
    \begin{pmatrix}
        s & p & s\\
        -m_s & y & m_s'
    \end{pmatrix}
    |sm_s\rangle \langle s m_s'|
    \label{eq:spherical_spin_operator}
\end{equation}
with 
\begin{equation}
n_p = \frac{1}{\sqrt{(p+2)!}}.
\end{equation}
The normalization is chosen, such that
\begin{equation}
\begin{aligned}
    \sigma^0_0 &= \mathbb{I},\\
    \sigma^1_0 &= \sigma_z,\\
    \sigma^1_{\pm 1} &= \mp \frac{1}{\sqrt{2}}(\sigma_x \pm i\sigma_y)
    \label{eq:transform_pauli_to_spherical}
\end{aligned}
\end{equation}
as shown in the appendix \ref{app:derivation_spinopangular}.
As can be derived from eq. \eqref{eq:transform_pauli_to_spherical}, Hermitian conjugation of the spherical tensor components gives
\begin{equation}
(\sigma^{p}_y)^\dagger = (-1)^y \sigma^p_{-y}.
\label{eq:hermconj_spintensor}
\end{equation}
As the transformation to spherical tensors as defined in eq. \eqref{eq:transform_pauli_to_spherical} is unitary, the Hilbert-Schmidt inner product is preserved:
\begin{equation}
\Tr[(\sigma^p_y)^\dagger \sigma^{p'}_{y'}]=2\delta_{pp'}\delta_{yy'},
\label{eq:orthogonality_spin}
\end{equation}
and the spin tensor components defined in eq. \eqref{eq:spherical_spin_operator} thus form an orthogonal basis for the spin-\(\frac12\) operator space \(\mathrm{End}(V_s)\cong V_s\otimes V_s^*\) over the complex numbers.

\subsection{Coupling of orbital and spin operators to multipoles}
We now complete the construction of the multipole basis of the operator space given in eq.~\eqref{eq:reshuffled_operator_space} by considering the tensor product between the orbital and spin operators. Analogously to the discussion for the orbital tensors, the products $C^k_q(l',l)\otimes \sigma^p_y$ are themselves a representation of $\mathrm{SO}(3)$ and can thus be decomposed into irreducible components, i.e. the multipoles
\begin{equation}
\begin{aligned}
T^r_t &= [C^k_q(l',l) \otimes \sigma^p_y]^r_t\\
&=\sum_{q,y} n_{\mathrm{CG}}^{-1}\langle k q ; p y | r t \rangle  C^k_q(l',l) \otimes \sigma^p_y\\
&=\sum_{q,y} n_{3j}^{-1}(-1)^{r-t}
\begin{pmatrix}
    r & p & k \\
    -t & y & q
\end{pmatrix}
 C^k_q(l',l) \otimes \sigma^p_y
\end{aligned}
\end{equation}
with $r = |k-p|,\dots,k+p$. 
This formation of the tensor product is often referred to as coupling of the orbital and spin operators.
More naturally, we can rewrite the operator as:
\begin{equation}
\begin{aligned}
T^r_t=&n_{3j}^{-1} (-1)^{k+p}\\
& \times \sum_{q,y} (-1)^{-q-y}
\begin{pmatrix}
    k & r & p \\
    -q & t & -y
\end{pmatrix}
 C^k_q(l',l) \otimes \sigma^p_y,
 \label{eq:multipole_realnorm}
\end{aligned}
\end{equation}
using the symmetry properties of the Wigner-3j symbols and $(-1)^{2r}=1$ for $r$ an integer, as well as conservation of quantum number $t=q+y$. Eq. \eqref{eq:multipole_realnorm} is the form used by van der Laan \emph{et al.} \cite[Eq.~A1]{vanderLaan1995} and Bultmark \emph{et al.} \cite[Eq.~26]{Bultmark2009}, but here our orbital tensor operator is for general $l$ and $l'$.
As here we are coupling two tensor operators, rather than normalizing expectation values for certain states as was the case for the orbital tensor operator, choosing an analogous normalization to eq.~\eqref{eq:orbital_normalization} like \(\begin{pmatrix}
    k & r & p\\
    -k & k-p & p
\end{pmatrix}\),
would not have a clear interpretation.
Omitting any normalization would correspond to treating the coupling as an addition of angular momenta, which is normalized for the maximal momentum, i.e. $T^r_{q+p} = C^k_k(l',l)\otimes\sigma^p_p$. 
For the coupled tensor, we instead choose the normalization such that the $q=0,y=0$ term in $T^r_0$ appears with unit coefficient, namely as $C^k_0(l',l)\otimes\sigma^p_0$. With this convention, the zero component is tied to the $z$-components of the orbital and spin tensors.
For $k+p+r$ even, we can achieve such a normalization by dividing by the projecting factor
\(\begin{pmatrix}
    k & p & r\\
    0 & 0 & 0
\end{pmatrix}\).

Edmonds \cite[Eq.~3.7.17]{Edmonds1957} derived the following explicit expression:
\begin{equation}
\begin{aligned}
\begin{pmatrix}
    a & b & c\\
    0 & 0 & 0
\end{pmatrix}
= (-1)^{g/2} &\sqrt{\frac{(g-2a)!(g-2b)!(g-2c)!}{(g+1)!}}\\
\times &\frac{(g/2)!}{(g/2-a)!(g/2-b)!(g/2-c)!},
\end{aligned}
\label{eq:3jsymbol_kpr}
\end{equation}
for $g=a+b+c$ even.
One can extend this expression to odd $g$ by allowing for half-integer factorials, which van der Laan \emph{et al.} \cite{vanderLaan1995} achieved with the double factorials $k!!$, i.e. the product of all integer numbers up to $k$ of same parity as $k$. For $k$ even, the relation $(2k)!!=2^k k!$
gives $(k/2)!= 2^{-k/2}k!!$. As $2^{-k/2}k!!$ is non-zero also for $k$ odd, it can be substituted in the right hand side of eq. \eqref{eq:3jsymbol_kpr}
\begin{equation}
\begin{aligned}
(-1)^{g/2} &\sqrt{\frac{(g-2a)!(g-2b)!(g-2c)!}{(g+1)!}}\\
\times &\frac{g!!}{(g-2a)!!(g-2b)!!(g-2c)!!},
\end{aligned}
\end{equation}
which now is non-zero also for $g$ odd and coincides with the Wigner-3j symbol for $g$ even.
The factor $(-1)^{g/2}$ can be replaced by $i^g$. This factor conveniently removes the $k,p,$ and $r$ dependent sign in Hermitian conjugation relations, see appendix \ref{app:hermconj_multipole}, and for $l'=l$ results in the spherical tensor components satisfying eq. \eqref{eq:spherical-hermitian}, see subsection \ref{subsec:real}.
Putting everything together we arrive at the following expression for the multipole operator for general $l$ and $l'$:

\begin{equation}
\begin{aligned}
W^{kpr}_t(l',l)=&n_{kpr}^{-1} (-1)^{k+p}\\
& \times \sum_{q,y} (-1)^{-q-y}
\begin{pmatrix}
    k & r & p \\
    -q & t & -y
\end{pmatrix}
 C^k_q(l',l) \otimes \sigma^p_y,
 \label{eq:multipole_tensor}
\end{aligned}
\end{equation}
with
\begin{equation}
\begin{aligned}
    n_{kpr} = i^g &\sqrt{\frac{(g-2a)!(g-2b)!(g-2c)!}{(g+1)!}}\\
\times &\frac{g!!}{(g-2a)!!(g-2b)!!(g-2c)!!},
\end{aligned}
\end{equation}
with $g=k+p+r$ and $a=k$, $b=p$ and $c=r$.
The expression coincides with the multipole tensor as defined by van der Laan \emph{et al.} \cite{vanderLaan1995} for $l'=l$. The Hermitian adjoint of the components is
\begin{equation}
   \left(W^{kpr}_t(l',l)\right)^\dagger = (-1)^{l-l'+t} W^{kpr}_{-t}(l,l'),
\label{eq:hermconj_multipoletensor} 
\end{equation}
as derived in the appendix \ref{app:hermconj_multipole}

Orthogonality follows from orthogonality of the orbital and spin tensor components; it is computed explicitly in the appendix \ref{app:orthogonality_sphericalmultipole}. Completeness is shown in the appendix \ref{app:orthogonality_completeness}.
The spherical multipole operators thus form a complete orthogonal basis of the whole operator space \(\mathrm{End}(\mathcal{H})\).

\subsection{Real Hermitian operator basis}
\label{subsec:real}

For an operator to represent a physical observable, it has to be Hermitian, and given an operator $O$, Hermitian operators can be constructed by taking the following linear combinations:
\begin{equation}
    \begin{aligned}
        O^+ &= \frac{1}{2}\left(O+O^\dagger\right),\\
        O^- &= \frac{1}{2i}\left(O-O^\dagger\right).
    \end{aligned}
    \label{eq:transform_to_hermitian}
\end{equation}

Another property we would like to achieve concerns how the tensor operator transforms under rotations. The spherical tensor components $T^r_t$ transform under rotation $R\in\mathrm{SO}(3)$ via the Wigner-D matrices \cite[Eq.~5.2.1]{Edmonds1957}:
\begin{equation}
U(R) T^r_t U^\dagger(R)
=
\sum_{t'=-r}^{r}
D^{(r)}_{t't}(R) T^r_{t'} ,
\label{eq:rotation_of_sphericaltensorcomponents}
\end{equation}
and, in particular, for a rotation around the $z$ axis by an angle $\alpha$,
\begin{equation}
U(R_z(\alpha)) T^r_t U^\dagger(R_z(\alpha))
=
e^{-it\alpha} T^r_t ,
\label{eq:z_rotation_spherical}
\end{equation}
up to the sign convention chosen for active or passive rotations \cite[Chp.~4.16,~Eq.~2]{varshalovich1988}. 
One can observe that this spherical basis diagonalizes rotations.
However, we would like the rotation matrices $U(R)$ to be real orthogonal matrices. This can be achieved by taking the appropriate linear combinations of the tensor operator components to bring the tensor operator into its so-called real form \cite[Eq.~8.5,~8.6]{Choi1999}:
\begin{equation}
\begin{aligned}
    &T^{r,\mathrm{real}}_0 = T^r_0, & \text{for }t=0,\\
    &
    \begin{cases}
        T^{r,\mathrm{real,+}}_t &= \frac{1}{\sqrt{2}}\left(T^r_t + (-1)^t T^r_{-t} \right),\\
        T^{r,\mathrm{real,-}}_t &= \frac{-i}{\sqrt{2}}\left(T^r_t -(-1)^tT^r_{-t} \right), \\
    \end{cases}
    & \text{for }t\neq0.\\
\end{aligned}
\label{eq:transform_to_real_paper}
\end{equation}
The factor $(-1)^t$ is a phase convention inherited by the Condon-Shortley phase convention. Coupling using Clebsch-Gordan/Wigner-3j symbols, keeps these relations intact \cite[Eq.~5.1.7]{Edmonds1957}; global phase factors like $i^{k+p+r}$ do not introduce a $t$-dependent phase. Note that the positive and negative components of these real operators are not independent, and one can define them as follows
\cite{BLANCO1997}:
\begin{equation}
    \begin{aligned}
        T^{r,\mathrm{real}}_t =
        &
        \begin{cases}
            \frac{(-1)^t}{\sqrt{2}}\left( T^r_t + (-1)^t T^r_{-t}\right), &\quad t>0,\\
            T^r_0, &\quad t=0,\\
            \frac{(-1)^t}{i\sqrt{2}}\left(T^r_{|t|} - (-1)^t T^r_{-|t|}  \right), &\quad t<0,\\
        \end{cases}
    \end{aligned}
    \label{eq:transform_to_real}
\end{equation}

If the spherical tensor components satisfy the following Hermitian conjugation relation:

\begin{equation}
    \left( T^r_t \right)^\dagger = (-1)^t T^r_{-t},
    \label{eq:spherical-hermitian}
\end{equation}

the transformations to Hermitian operators in eq. \eqref{eq:transform_to_hermitian} and so-called real operators in eq. \eqref{eq:transform_to_real_paper}, or rather \eqref{eq:transform_to_real}, coincide up to global sign and normalization conventions.

For $l'=l$, the multipole operators $W^{kpr}_t(l',l)$ defined in eq. \eqref{eq:multipole_tensor}, satisfy eq. \eqref{eq:spherical-hermitian}, and the construction as defined in eq. \eqref{eq:transform_to_real} results in real Hermitian multipoles.
For $l'\neq l$, however, the spherical multipole operators have the following Hermitian conjugation relation:
\begin{equation}
    \left(W^{kpr}_t(l',l)\right)^\dagger = (-1)^{l-l'+t} W^{kpr}_{-t}(l,l'),
\end{equation}
as derived in the appendix \ref{app:hermconj_multipole}, and a transformation as in eq.~\eqref{eq:transform_to_real}, does not result in both real and Hermitian multipoles.

One solution could be to introduce the factor $i^{l'+l}$ in the normalization of the orbital operators. Then, Hermitian conjugation would only raise the factor $(-1)^t$ due to the Condon-Shortley phase convention:
\begin{equation}
\left(\tilde{W}^{kpr}_t(l',l)\right)^\dagger = (-1)^t\tilde{W}^{kpr}_{-t}(l,l'),
\end{equation}
and the Hermitian multipoles constructed via eq. \eqref{eq:transform_to_hermitian} would be real multipoles as well, as the transformation under $\mathrm{SO}(3)$ is independent of $l'$ and $l$ (Eq. \eqref{eq:rotation_of_sphericaltensorcomponents}). Time-reversal symmetry, as discussed in subsection \ref{sec:decomp_time} would further resolve the relation between the $t$ and $-t$ components within the same $(l',l)$ block of such $\tilde{W}^{kpr,\mathrm{real}}_t(l',l)$ and allow reduced indexing analogous to eq. \eqref{eq:transform_to_real}.
As discussed in subsection \ref{subsec:decomp_orbital_tensor}, the phase factor $i^{l'+l}$ would also change the normalization convention for $l'=l$.

To stay consistent with the normalization for $l'=l$ we do not introduce the factor $i^{l'+l}$. We construct real Hermitian multipoles by performing both transformations to real \eqref{eq:transform_to_real} and Hermitian \eqref{eq:transform_to_hermitian} operators in turn.

We first perform the transformation to real components:
\begin{equation}
    \begin{aligned}
        &W^{kpr,\mathrm{real}}_t(l',l)\\ =
        &
        \begin{cases}
            \frac{(-1)^t}{\sqrt{2}}\left( W^{kpr}_t(l',l) +  (-1)^tW^{kpr}_{-t}(l',l)\right), & t>0,\\
            W^{kpr}_0(l',l), & t=0,\\
            \frac{(-1)^t}{i\sqrt{2}}\left( W^{kpr}_{|t|}(l',l) - (-1)^tW^{kpr}_{-|t|}(l',l)\right), & t<0,\\
        \end{cases}
    \end{aligned}
    \label{eq:real_multipoles}
\end{equation}

These real multipole components now transform under Hermitian conjugation as
\begin{equation}
    \begin{aligned}
        \left(W^{kpr,\mathrm{real}}_t(l',l) \right)^\dagger = (-1)^{l-l'} W^{kpr,\mathrm{real}}_t(l,l'),
    \end{aligned}
    \label{eq:hermconj_realmultipoles}
\end{equation}
Next we construct Hermitian multipole operators:

\begin{equation}
    \begin{aligned}
        &W^{kpr,\mathrm{real},+}_t(l',l)\\ &= \frac{1}{2}\left(W^{kpr,\mathrm{real}}_t(l',l) + (-1)^{l-l'}W^{kpr,\mathrm{real}}_t(l,l')\right),\\
        &W^{kpr,\mathrm{real},-}_t(l',l) \\&= \frac{1}{2i}\left(W^{kpr,\mathrm{real}}_t(l',l) - (-1)^{l-l'}W^{kpr,\mathrm{real}}_t(l,l')\right),\\
    \end{aligned}
    \label{eq:hermitian_multipoles}
\end{equation}

For these real Hermitian multipole components, there is a clear relation between the operators of opposite blocks $W^{kpr,\mathrm{real},\pm}_t(l',l)$ and $W^{kpr,\mathrm{real},\pm}_t(l,l')$. It thus makes sense to stick to an indexing convention like $l'\leq l$.

\section{Decomposition by Parity and Time-Reversal Symmetry}
\label{sec:decomp_parity_time}
Physical systems are often characterized by their spatial parity and time-reversal symmetry. Decomposing multipole operators according to these symmetries makes it possible to resolve their contributions at the level of the local electronic structure.

\subsection{Decomposition by parity}
The parity operator $P$ acts on orbital and spin states as
\begin{equation}
P|l,m\rangle = (-1)^l |l,m\rangle, 
\qquad 
P|s,m_s\rangle = |s,m_s\rangle.
\end{equation}
It follows that orbital operators transform as
\begin{equation}
P\, C^k_q(l',l)\, P^{-1} = (-1)^{l + l'} C^k_q(l',l),
\end{equation}
while spin tensor operators are invariant,
\begin{equation}
P\, \sigma^p_\mu \, P^{-1} = \sigma^p_\mu.
\end{equation}
Accordingly, the coupled tensor operators transform as
\begin{equation}
P\, [C^k_q(l',l)\otimes \sigma^p_\mu]^r_t \, P^{-1} 
= (-1)^{l + l'} [C^k_q(l',l)\otimes \sigma^p_\mu]^r_t.
\end{equation}

Thus, the parity of the coupled tensor operators, and hence of the multipoles, 
is determined by $(-1)^{l + l'}$:
\begin{equation}
 P W^{kpr\pm}_t(l',l) P^{-1} = (-1)^{l'+l} W^{kpr\pm}_t(l',l)
\end{equation}
In particular, operators connecting states 
with the same orbital angular momentum ($l=l'$) are even under parity, while 
operators connecting different $l$ sectors can have either even or odd parity.

\subsection{Decomposition by time-reversal symmetry}
\label{sec:decomp_time}
Next we decompose the multipole tensor operators by time-reversal symmetry.
As discussed in the appendix \ref{app:timereversal}, 
we consider the projection of an operator $A\in \mathrm{End}(\mathcal{H})$ onto its time-reversal even (\(\nu=0\)) and time-reversal odd (\(\nu=1\)) components:
\begin{equation}
A^{\nu}=P_{\nu}(A) = \frac{1}{2}(A + (-1)^{\nu}\Theta A \Theta^\dagger),
\end{equation}
where $\Theta A \Theta^\dagger$ is the time-reversed operator A, as determined in appendix \ref{app:timereversed_operators}, which we abbreviate by $\tau(A)$.

In the appendix \ref{app:timerev_sphericalmultipoleoperators}, we derive the time-reversed spherical multipole tensor components:
\begin{equation}
\tau \left(W^{kpr}_t(l',l)\right) = (-1)^{l'+l+k+p+t} W^{kpr}_{-t}(l',l).
\end{equation}

For our linear combinations within each $(l',l)$ block stated in eq. \eqref{eq:real_multipoles}, we obtain
\begin{equation}
    \tau\left(W^{kpr,\mathrm{real}}_t(l',l) \right) = (-1)^{l'+l+k+p}W^{kpr,\mathrm{real}}_t(l',l),
\end{equation}

and for our Hermitian multipole operators:
\begin{equation}
\begin{aligned}
        \tau\left(W^{kpr,\mathrm{real},+}_t(l',l)\right) &= (-1)^{l'+l+k+p}W^{kpr,\mathrm{real},+}_t(l',l),\\
        \tau\left(W^{kpr,\mathrm{real},-}_t(l',l)\right) &= (-1)^{l'+l+k+p+1}W^{kpr,\mathrm{real},-}_t(l',l),
\end{aligned}
\end{equation}
which enforces the following rule for trivial zero-valuedness:

\begin{table}[H]
    \centering
    \renewcommand{\arraystretch}{1.2}
    \begin{tabular}{ccc}
        \toprule
         & $W^{kpr,\mathrm{real},+,\nu}_t(l',l)$ & $W^{kpr,\mathrm{real},-,\nu}_t(l',l)$\\
        \midrule
        $l'+l+k+p+\nu$ even & $\neq0$ & $0$\\
        $l'+l+k+p+\nu$ odd  & $0$ & $\neq0$\\
        \bottomrule
    \end{tabular}
    \caption{
        Selection rules for the time-reversal even/odd Hermitian multipole operators $W^{kpr,\mathrm{real},\pm,\nu}_t(l',l)$.
        Depending on the parity of $l'+l+k+p+\nu$, the operator vanishes trivially.
    }
    \label{tab:hermitian_time_selection_rule}
\end{table}
Note that for $l'=l$ multipoles, only `$+$' multipoles are constructed, and the second row in this categorization in Table \ref{tab:hermitian_time_selection_rule} is never accessed, i.e. no time-reversal even/odd multipoles can be constructed with $k+p$ odd/even for $l'=l$.
It is also interesting to note that most of the multipoles of interest fall into the first row in Table \ref{tab:hermitian_time_selection_rule}, for example the charge, magnetic and magnetoelectric multipoles and many orbital operators like orbital dipole, spin-orbit coupling and the electric toroidal dipole.

With this definite relation between $l'+l+k+p+\nu$ and $\pm$, we can drop the $\pm$ label:
\begin{equation}
\begin{aligned}
    W^{\nu kpr, \mathrm{real}}_{l_1, l_2, t} &=
    \begin{cases}
        W^{kpr,\mathrm{real},+}_t(l',l), & l'+l+k+p+\nu \text{ even},\\
        W^{kpr,\mathrm{real},-}_t(l',l), & l'+l+k+p+\nu \text{ odd},\\
    \end{cases}
\end{aligned}
\label{eq:timesymmetrized_multipoles}
\end{equation}
using $l'=\mathrm{min}\{l_1, l_2\}$ and $l=\mathrm{max}\{l_1, l_2\}$ for definite ordering. We call $W^{\nu kpr, \mathrm{real}}_{l_1, l_2, t}$ time-reversal even and odd multipoles.
These multipoles form a complete orthogonal basis set consisting of Hermitian operators for the space $\mathrm{End}(\mathcal{H})$, providing a basis for the local density matrix including all inter-shell components, as shown in the appendix \ref{app:orthogonality_completeness} and \ref{app:multipole_decomp}. Time-reversal symmetry introduces an additional relation between different expectation values of the spherical multipole operators, i.e. spherical multipole moments, which recover the selection rules in Table \ref{tab:hermitian_time_selection_rule}, as shown in the appendix \ref{app:additional_relation}.

\section{A case study: Chirality in Tellurium}
\label{sec:casestudy}

One newly accessible multipole in our extension to inter-shell operators is the parity-odd, time-reversal even electric toroidal monopole (ETM), which has been proposed as a suitable order parameter for chirality~\cite{Oiwa_PRL:2022}. In our notation, the ETM corresponds to the operator
\begin{equation}
W^{\nu=0,110}_{l_1,l_2,0} \quad \text{with} \quad l_1+l_2 \text{ odd}.
\end{equation}
This is a scalar quantity that is insensitive to time-reversal, and as a monopole it is invariant under proper rotations. However, since it is parity-odd, it changes sign under inversion and mirror operations and so it is chiral.

An operator of identical symmetries is $r\cdot (l\times s)$, which has been previously discussed in the literature \cite{Kusunose_chirality_APL_2024}.
Cyclic permutation to $s \cdot (r \times l)$ reveals that it is of rank 1 in both spin and orbital parts.
When promoting the expression to a quantum mechanical operator, the form required is
\begin{equation}
\frac{1}{2}\left(\hat r\cdot (\hat l\times \hat s) + (\hat l\times \hat s) \cdot \hat r \right),
\label{eq:ETM_operator}
\end{equation}
to ensure physical meaning, i.e. Hermiticity,
since $r$ and $l$ do not commute.

For the operator $\hat r$, the only non-zero matrix elements are between orbital spaces $l_1$ and $l_2$ differing by 1, consistent with the restriction of the orbital rank $k$ being in the range $|l_1-l_2|$ to $l_1 + l_2$, and the requirement of $l_1 + l_2$ odd for parity-odd operators.

Due to orthogonality of the multipole basis, we can now identify that the only multipole moments of the density matrix, denoted by lowercase letters, see eq. \eqref{eq:multipole_moments}, that contribute to the expectation value of the operator eq. \eqref{eq:ETM_operator} are the electric toroidal monopole moments
\begin{equation}
w^{\nu=0,110}_{l_1,l_2,0} \quad \text{with} \quad l_2=l_1+1.
\end{equation}

Next, we show for the prototypical chiral material tellurium that this ETM, accessible with our multipole decomposition, indeed acts as an atomic-scale measure for structural chirality.

\subsection{Computational details}
To obtain the density matrix that we decompose into irreducible spherical tensor components, we employ density functional theory (DFT) as implemented in the Vienna ab initio simulation package (VASP)~\cite{Kresse:1996, KresseJoubert1999PAW}.
We choose a $\Gamma$-centered $7\times7\times6$ $k$-point grid and an energy cutoff of $350\,\mathrm{eV}$ for the plane-wave basis, which converge the total energy to within $1\,\mathrm{meV}$ per atom.
We use the Perdew-Burke-Ernzerhof (PBE) exchange-correlation functional~\cite{PBE:1996} and van der Waals forces according to the Tkatchenko-Scheffler method with iterative Hirshfeld partitioning~\cite{IVDW21:2013, IVDW21:2014}. 
We relax the tellurium structure with $P3_121$ space group until the forces acting on each atom are $<0.1\,\mathrm{meV}/\text{\AA}$. The relaxed lattice constants $a=b=4.43\,\text{\AA}$ and $c=5.92\,\text{\AA}$ differ by less than $0.6\%$ from the experimental values
~\cite{Keller_Te:1977}.
We then invert the optimized structure to obtain the opposite enantiomer belonging to the $P3_221$ space group (see Fig.~\ref{fig_Te_ETM}).
We include spin-orbit coupling in the final self-consistent DFT calculation to obtain the density matrix for extracting the multipole moments.

\subsection{Results}
We extract the atomic-site ETM moments $w^{\nu=0,110}$ with $(l_1,l_2)=(s,p)$ for both tellurium enantiomers from our multipole decomposition of the density matrix and report the values in Table~\ref{tab:etm}. Within one enantiomer, the ETM moment is the same on each of the three tellurium atoms in the unit cell. In the opposite enantiomer, the ETMs are again equal to each other, with the same size but opposite sign to the first enantiomer. We visualize this ETM pattern in the lower part of Fig.~\ref{fig_Te_ETM} by blue inward-pointing (red outward-pointing) hedgehogs for negative (positive) local monopole moments.
To conclude, our calculated ETM moments act as a local measure of chirality since they allow to differentiate the two enantiomers by assigning amplitudes on atomic sites that are equal in value but opposite in sign for the two enantiomers differing by the rotational sense of the 3-fold screw axis.
\begin{figure}[t!]
\includegraphics[trim = 80mm 20mm 60mm 10mm, clip, width=1\columnwidth]{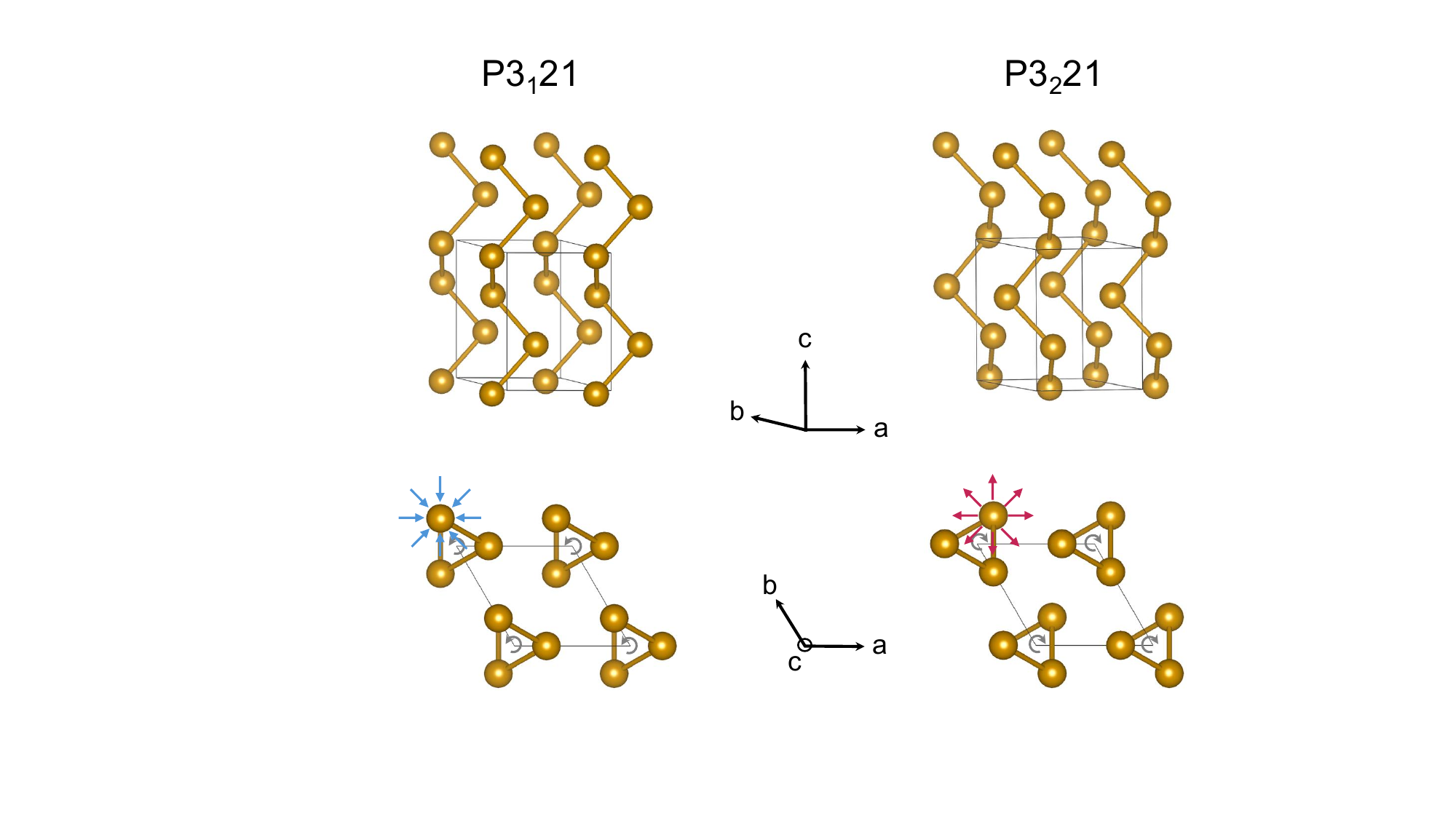}
\caption[]{Crystal structure of chiral tellurium visualized with \textsc{vesta}~\cite{VESTA_Momma:2008} with the $c$ axis vertical (upper part) and out of the plane (lower part). The curly gray arrows indicate the sense of rotation of the Te chains around the $c$ axis. 
The blue inward- and red outward-pointing arrows around one Te atom schematically depict the computed atomic ETM moments that are equal on each Te site within one enantiomer and switch sign between enantiomers.
}
\label{fig_Te_ETM}
\end{figure}
\begin{table}[H]
    \centering
    \begin{tabular}{cc|c}
    \toprule
    $w^{\nu=0,110}_{s,p,0}$ & $P3_121$ & $P3_221$\\
    \midrule
     Site 1 &   $-2.6 \times 10^{-4}$ &  $2.6 \times 10^{-4}$\\
     Site 2 &   $-2.6 \times 10^{-4}$ &  $2.6 \times 10^{-4}$\\
     Site 3 &   $-2.6 \times 10^{-4}$ &  $2.6 \times 10^{-4}$\\
    \bottomrule
    \end{tabular}
    \caption{Calculated $w^{\nu=0,110}_{s,p,0}$ for the three Te sites in the unit cell of tellurium.}  
\label{tab:etm}
\end{table}

\section{Summary}
In summary, we have completed the single-site multipole formalism by extending it with the missing inter-shell operators and demonstrated its applicability by calculating the electric toroidal monopoles for the two enantiomers of tellurium.
We revisited the formalism for the decomposition of the density matrix into multipole moments by starting with the space of endomorphisms acting on the single-particle Hilbert space of an electron localized on a given site.
By analyzing its tensor product structure, we separated orbital and spin operator subspaces, which enabled us to separately decompose each subspace by the group action of $\mathrm{SO}(3)$. Keeping track of orbital quantum numbers, we obtained a basis of the orbital operator space, and with subsequent coupling to the spin operators, we constructed a complete orthogonal set of spherical multipole operators, consistent with existing fixed-shell formulations.\\
We then transformed the spherical multipole tensors into a real Hermitian operator basis and identified their symmetries under parity and time-reversal.
With these multipoles and their classification, we then decomposed the density matrix as obtained from DFT calculations into symmetry-labeled amplitudes, providing a natural physically interpretable representation of the local electronic system. As many physical operators have definite spatial and time-reversal symmetry, we can identify which will give a non-zero expectation value and have a metric to compare them between different systems and sites.\\
Finally, we demonstrated that the newly accessible electric toroidal monopole is equal in value and opposite in sign on the atoms in the two enantiomers of tellurium, suggesting that it provides a local atomic measure of chirality.

\section*{Acknowledgments}
This work was supported by the Swiss National Science Foundation (SNSF) under Grant Nos. 10002603 and 225790.
Computational resources were provided by the ETH Zurich Euler cluster. 

\section*{Data availability}
The relevant input files of our ab initio calculations and the data supporting the findings of this work will be publicly available upon publication of this work.  \vfill\eject

\appendix

\section{Explicit derivations}

\subsection{Derivation of the spherical spin tensor components in angular momentum basis}
\label{app:derivation_spinopangular}
For \(s=\frac12\), the spin-operator space decomposes as
\begin{equation}
V_s\otimes V_s^* \simeq \frac12\otimes \frac12 = 0\oplus 1 .
\end{equation}
Hence the four spin operators may be organized into irreducible spherical tensors \(\sigma^p_y\), with \(p=0,1\):
\begin{equation}
\sigma^p_y
=
\sum_{m,m'}
n_p^{-1}
(-1)^{s-m}
\begin{pmatrix}
s & p & s\\
-m & y & m'
\end{pmatrix}
|s m\rangle\langle s m'|.
\end{equation}
We show that the spherical spin operators as given in eq.~\eqref{eq:spherical_spin_operator} are reproduced by choosing the following normalization constant:
\begin{equation}
n_p=\frac{1}{\sqrt{(p+2)!}},
\end{equation}
which for $p=0,1$ is \(n_0^{-1}=\sqrt2\) and \(n_1^{-1}=\sqrt6\).

For \(p=0\), using
\begin{equation}
\begin{pmatrix}
s & 0 & s\\
-m & 0 & m'
\end{pmatrix}
=
\frac{(-1)^{s-m}}{\sqrt{2s+1}}\delta_{m m'} ,
\end{equation}
with \(s=\frac12\), one obtains
\begin{equation}
\sigma^0_0
=
\sum_m |s m\rangle\langle s m|
=
\mathbb I .
\end{equation}
Thus the scalar component is normalized as the identity.

For \(p=1\), direct evaluation of the Wigner-3j symbols gives:
\begin{equation}
\begin{aligned}
\begin{pmatrix}
\frac12 & 1 & \frac12\\
-\frac12 & 0 & \frac12
\end{pmatrix}
=
\begin{pmatrix}
\frac12 & 1 & \frac12\\
\frac12 & 0 & -\frac12
\end{pmatrix}
=
\frac{1}{\sqrt6},
\\
\begin{pmatrix}
\frac12 & 1 & \frac12\\
-\frac12 & 1 & -\frac12
\end{pmatrix}
=
-\frac{1}{\sqrt3},
\\
\begin{pmatrix}
\frac12 & 1 & \frac12\\
\frac12 & -1 & \frac12
\end{pmatrix}
=
-\frac{1}{\sqrt3},
\end{aligned}
\end{equation}

which gives with
\begin{equation}
|+\rangle\equiv \left|\frac12,\frac12\right\rangle,
\qquad
|-\rangle\equiv \left|\frac12,-\frac12\right\rangle ,
\end{equation}
the components
\begin{equation}
\sigma^1_0
=
|+\rangle\langle+|-|-\rangle\langle-|
=
\sigma_z ,
\end{equation}
and
\begin{equation}
\sigma^1_{+1}
=
-\sqrt2\,|+\rangle\langle-| ,
\qquad
\sigma^1_{-1}
=
\sqrt2\,|-\rangle\langle+| .
\end{equation}
Since
\begin{equation}
\sigma_x=|+\rangle\langle-|+|-\rangle\langle+|,
\qquad
\sigma_y=-i|+\rangle\langle-|+i|-\rangle\langle+| ,
\end{equation}
these may equivalently be written as
\begin{equation}
\sigma^1_{\pm1}
=
\mp \frac{1}{\sqrt2}
\left(\sigma_x\pm i\sigma_y\right).
\end{equation}
Thus the chosen normalization reproduces the identity and the Pauli matrices in spherical form.
The same explicit expressions also give
\begin{equation}
(\sigma^0_0)^\dagger=\sigma^0_0,
\qquad
(\sigma^1_0)^\dagger=\sigma^1_0,
\qquad
(\sigma^1_{\pm1})^\dagger=-\sigma^1_{\mp1}.
\end{equation}
Therefore, for both \(p=0\) and \(p=1\),
\begin{equation}
(\sigma^p_y)^\dagger
=
(-1)^y\sigma^p_{-y}.
\end{equation}

\subsection{Hermitian adjoint of the spherical orbital tensor components}
\label{app:hermconj_orbitaltensor}
We compute the Hermitian adjoint of the spherical orbital tensor components as defined in eq. \eqref{eq:orbital_tensor_operator}:
\begin{equation}
\begin{aligned}
&\left(C^k_q(l',l)\right)^\dagger\\
&=
\sum_{m,m'}
n_{lkl'}^{-1}
(-1)^{l-m}
\begin{pmatrix}
l & k & l'\\
-m & q & m'
\end{pmatrix}
|l'm'\rangle\langle lm| \\
&=
\sum_{m,m'}
n_{lkl'}^{-1}
(-1)^{l-m}
\begin{pmatrix}
l' & k & l\\
-m' & -q & m
\end{pmatrix}
|l'm'\rangle\langle lm|,
\end{aligned}
\end{equation}
where in the second line we used the symmetry of the Wigner-3j symbol (switching the sign of the magnetic moments and permuting the columns).

Since the Wigner-3j symbol enforces
\begin{equation}
-m+q+m'=0,
\qquad\text{so that}\qquad m=q+m',
\end{equation}
we can rewrite
\begin{equation}
(-1)^{l-m}
=
(-1)^{l-l'-q}(-1)^{l'-m'},
\end{equation}

which, using $n_{lkl'}=n_{l'kl}$, results in
\begin{equation}
\begin{aligned}
&\left(C^k_q(l',l)\right)^\dagger\\
&= (-1)^{l-l'-q}\\
& \hspace{1em} \times
\sum_{m,m'}
n_{l'kl}^{-1}
(-1)^{l'-m'}
\begin{pmatrix}
l' & k & l\\
-m' & -q & m
\end{pmatrix}
|l'm'\rangle\langle lm| \\
&=
(-1)^{l-l'-q} C^k_{-q}(l,l'),
\end{aligned}
\end{equation}

\subsection{Hermitian adjoint of the spherical multipole operators}
\label{app:hermconj_multipole}

We first compute the Hermitian adjoint of the components of the spherical multipole tensor as defined in equation \eqref{eq:multipole_realnorm} with a real normalization constant:
\begin{equation}
\begin{aligned}
\left(T^r_t(l',l)\right)^\dagger
=&
n_{3j}^{-1} (-1)^{k+p}
\sum_{q,y} (-1)^{-q-y}
\begin{pmatrix}
    k & r & p \\
    -q & t & -y
\end{pmatrix}\\
& \times
\left(C^k_q(l',l)\right)^\dagger
\otimes
\left(\sigma^p_y\right)^\dagger.
\end{aligned}
\end{equation}
Using the Hermitian conjugate of the orbital and spin tensors given in eq.~\eqref{eq:hermconj_orbitaltensor} and \eqref{eq:hermconj_spintensor}, this is
\begin{equation}
...=
n_{3j}^{-1} (-1)^{k+p+l-l'}
\sum_{q,y}
\begin{pmatrix}
    k & r & p \\
    -q & t & -y
\end{pmatrix}
C^k_{-q}(l,l')
\otimes
\sigma^p_{-y},
\end{equation}
where we used that $q$ is an integer and thus $(-1)^{-2q}=1$.
Now relabeling the summation indices $q\to-q$, $y\to-y$ gives

\begin{equation}
...=
n_{3j}^{-1} (-1)^{k+p+l-l'}
\sum_{q,y}
\begin{pmatrix}
    k & r & p \\
    q & t & y
\end{pmatrix}
C^k_q(l,l')
\otimes
\sigma^p_y.
\end{equation}
Using the Wigner-3j symmetry we obtain
\begin{equation}
\begin{aligned}
&\left(T^r_t(l',l)\right)^\dagger=
n_{3j}^{-1} (-1)^{k+p+l-l'}
(-1)^{k+r+p}\\
& \hspace{2em}\times
\sum_{q,y}
\begin{pmatrix}
    k & r & p \\
    -q & -t & -y
\end{pmatrix}
C^k_q(l,l')
\otimes
\sigma^p_y.
\end{aligned}
\end{equation}

Since the Wigner-3j symbol enforces
\begin{equation}
-q-t-y=0,
\end{equation}
we have
\begin{equation}
(-1)^{-q-y}=(-1)^t.
\end{equation}

Therefore, introducing $1=(-1)^t(-1)^{-q-y}$, we obtain
\begin{equation}
\left(T^r_t(l',l)\right)^\dagger
=
(-1)^{k+p+r+l-l'}(-1)^t
T^r_{-t}(l,l').
\end{equation}

With the complex factor $i^{k+p+r}$ in the normalization constant, the Hermitian adjoint of the multipole tensor components, as defined in eq. \eqref{eq:multipole_tensor}, are:
\begin{equation}
   \left(W^{kpr}_t(l',l)\right)^\dagger = (-1)^{l-l'+t} W^{kpr}_{-t}(l,l'),
\end{equation}

\subsection{Orthogonality of operators}

\subsubsection{Orthogonality of the spherical orbital operators}
\label{app:orthogonality_orbital}
We determine the Hilbert-Schmidt inner product of the spherical tensor operator components as defined in eq. \eqref{eq:orbital_tensor_operator}.
Orthogonality of the spherical components of orbital tensors that maps between different $(l',l)$ blocks follows from the orthogonality of the $l$ orbital states. Orthogonality for orbital tensor components in the same $(l',l)$ block can be shown as follows:
\begin{equation}
\begin{aligned}
&\mathrm{Tr}\!\left[
\left(C^k_q(l',l)\right)^\dagger
C^{k'}_{q'}(l',l)
\right]\\
&= \mathrm{Tr}\!\left[(-1)^{l-l'-q}
\left(C^k_{-q}(l,l')\right)
C^{k'}_{q'}(l',l)
\right]
\end{aligned}
\end{equation}
Due to orthogonality of the orbital states and symmetry of the normalization constant $n_{l'kl}=n_{lkl'}$ this is:
\begin{equation}
\begin{aligned}
...=(-1)^{l-l'-q}\sum_{m,m'} &
n_{lkl'}^{-1}(-1)^{l'-m'}  
\begin{pmatrix}
l' & k & l\\
-m' & -q & m
\end{pmatrix}  \\
\times &n_{l'kl}^{-1}(-1)^{l-m} 
\begin{pmatrix}
l & k' & l'\\
-m & q' & m'
\end{pmatrix}
\end{aligned}
\end{equation}
Due to conservation of magnetic moment $-m'-q+m=-m+q'+m'=0$, we get $-q-m'-m=-2m$, hence with $m$ and $l$ being integers $(-1)^{l-l'-q+l'-m'+l-m}=(-1)^{2l-2m}=1$
\begin{equation}
...=n_{l'kl}^{-1}n_{lk'l'}^{-1}
\sum_{m,m'}
\begin{pmatrix}
l' & k & l\\
-m' & -q & m
\end{pmatrix}
\begin{pmatrix}
l & k' & l'\\
-m & q' & m'
\end{pmatrix}.
\end{equation}

We flip the signs on the first Wigner 3-j symbol and simultaneously switch its first and last column
\begin{equation}
...=n_{l'kl}^{-1}n_{lk'l'}^{-1}
\sum_{m,m'}
\begin{pmatrix}
l & k & l'\\
-m & q & m'
\end{pmatrix}
\begin{pmatrix}
l & k' & l'\\
-m & q' & m'
\end{pmatrix}.
\end{equation}

Using the orthogonality relation of Wigner-3j symbols, and the fact that $n_{l'kl}=n_{lkl'}$ we obtain
\begin{equation}
\mathrm{Tr}\!\left[
\left(C^k_q(l',l)\right)^\dagger
C^{k'}_{q'}(l',l)
\right]
=
\frac{\delta_{kk'}\delta_{qq'}}{n_{lkl'}^2 (2k+1)}
.
\label{eq:orthogonality_orbital}
\end{equation}
Thus the tensor operators, or rather their components, are mutually orthogonal with respect to the
Hilbert-Schmidt inner product.

\subsubsection{Orthogonality of the spherical multipole operators}
\label{app:orthogonality_sphericalmultipole}

We now show that the multipole tensor components as defined in eq. \eqref{eq:multipole_tensor} are orthogonal with respect to the Hilbert-Schmidt inner product. For different $k$ or $p$ values, orthogonality follows from orthogonality of the orbital and spin tensors, for different $r$ and $t$ values we have
\begin{align}
&\mathrm{Tr}\left[
\left(T^r_t(l',l)\right)^\dagger
T^{r'}_{t'}(l',l)
\right] .
\end{align}
Substituting the definition of the tensors gives
\begin{align}
&\mathrm{Tr}\left[
\left(W^{kpr}_t(l',l)\right)^\dagger
W^{kpr'}_{t'}(l',l)
\right]
\nonumber\\
&=
\overline{n_{kpr}^{-1}}n_{kpr'}^{-1}
\sum_{q,y}
\sum_{q',y'}
(-1)^{-q-y}
(-1)^{-q'-y'}\\
& \hspace{2em}\times
\begin{pmatrix}
    k & r & p \\
    -q & t & -y
\end{pmatrix}
\begin{pmatrix}
    k & r' & p \\
    -q' & t' & -y'
\end{pmatrix}
\nonumber\\
&\hspace{2em}\times
\mathrm{Tr}\left[
\left(C^k_q(l',l)\right)^\dagger
C^k_{q'}(l',l)
\right]
\mathrm{Tr}\left[
\left(\sigma^p_y\right)^\dagger
\sigma^p_{y'}
\right] .
\end{align}
Here the phase factor \((-1)^{k+p}\) drops out because it appears twice.

Using the orthogonality of the orbital and spin tensor components, see equations \eqref{eq:orthogonality_orbital} and \eqref{eq:orthogonality_spin},
we obtain
\begin{align}
...=&\overline{n_{kpr}^{-1}}n_{kpr'}^{-1}
\frac{2}{n_{l'kl}^2(2k+1)}\\
&\times
\sum_{q,y}
\begin{pmatrix}
    k & r & p \\
    -q & t & -y
\end{pmatrix}
\begin{pmatrix}
    k & r' & p \\
    -q & t' & -y
\end{pmatrix}.
\end{align}
The remaining sum is the standard orthogonality relation of Wigner-\(3j\) symbols,
\begin{equation}
\sum_{q,y}
\begin{pmatrix}
    k & r & p \\
    -q & t & -y
\end{pmatrix}
\begin{pmatrix}
    k & r' & p \\
    -q & t' & -y
\end{pmatrix}
=
\frac{\delta_{r r'}\delta_{t t'}}{2r+1}.
\end{equation}
Therefore,
\begin{equation}
\begin{aligned}
&\mathrm{Tr}\left[
\left(W^{kpr}_t(l',l)\right)^\dagger
W^{kpr'}_{t'}(l',l)
\right]\\
&\hspace{2em}=
\frac{2\delta_{r r'}\delta_{t t'}}{|n_{kpr}|^2(2r+1)n_{l'kl}^2(2k+1)}
 .
 \end{aligned}
\end{equation}

Including also orthogonality relations for $k$ and $p$, and the fact that $(l',l)$ index different operator spaces, which may be understood as different blocks of a bigger matrix operator, one obtains
\begin{equation}
\begin{aligned}
&\mathrm{Tr}\left[
\left(W^{kpr}_{t}(l',l)\right)^\dagger
W^{k'p'r'}_{t'}(\tilde l',\tilde l)
\right]\\
&\hspace{2em}=\frac{2\delta_{k k'}\delta_{p p'}
\delta_{r r'}\delta_{t t'}
\delta_{l \tilde l}\delta_{l' \tilde l'}}
{|n_{kpr}|^2(2r+1)n_{l'kl}^2(2k+1)}.
\end{aligned}
\end{equation}

\subsubsection{Orthogonality and completeness of the time even and odd multipole operators}
\label{app:orthogonality_completeness}
As shown in the appendix \ref{app:orthogonality_sphericalmultipole}, the spherical multipole operators are mutually orthogonal. As we have truncated the Hilbert space of the electron space $\tilde{\mathcal{H}}$ at some orbital quantum number $l_{max}$, the Hilbert space $\mathcal{H}$ considered is finite dimensional, and thus so is the operator space $\mathrm{End}(\mathcal{H})$. With the identification of the canonical basis elements $|u\rangle\langle v|$ with $|u\rangle \otimes\langle v|$ we have an isomorphism between the operator space $\mathrm{End}(\mathcal{H})$ and the tensor product space $\mathcal{H}\otimes \mathcal{H}^*$, as discussed in section \ref{sec:basis}. 
The Hilbert space $\mathcal{H}$ has dimension 
\begin{equation}
    \mathrm{dim}\mathcal{H}=\sum_{l=0}^{l_{max}} 2(2l+1)=2(l_{max}+1)^2,
\end{equation}
the factor 2 arises from the spin-$\frac12$ space and each orbital subspace has $2l+1$ orbital states.
Thus, the operator space has dimension $4(l_{max}+1)^4$.

We obtained the spherical tensor operators $W^{kpr}_t(l',l)$ by coupling orbital
operator of rank \(k\) with a spin operator of rank \(p\) to a total
operator of rank \(r\). The total number of operators can be computed as follows:

\begin{align}
    N
    &=
    \sum_{l',l=0}^{l_{\max}}
    \sum_{k=|l-l'|}^{l+l'}
    \sum_{p=0}^{1}
    \sum_{r=|k-p|}^{k+p}
    (2r+1) \\
    &=
    \sum_{l',l=0}^{l_{\max}}
    \sum_{k=|l-l'|}^{l+l'}
    \sum_{p=0}^{1}
    (2k+1)(2p+1) \\
    &=
    4
    \sum_{l',l=0}^{l_{\max}}
    \sum_{k=|l-l'|}^{l+l'}
    (2k+1) \\
    &=
    4
    \sum_{l',l=0}^{l_{\max}}
    (2l+1)(2l'+1) \\
    &=
    4
    \left(
    \sum_{l=0}^{l_{\max}} (2l+1)
    \right)^2 \\
    &=
    4(l_{\max}+1)^4,
\end{align}
where we have used
\begin{equation}
    \sum_{x=|n-m|}^{n+m}2x+1 = (2n+1)(2m+1),
\end{equation}
in the second and fourth equality. This equality can be shown using the Gaussian summation formula with assuming $m\leq n$, which due to the symmetry of the result in $n$ and $m$ gives the correct result.

The transformation from complex spherical tensor components to real tensor components within each fixed $(l',l)$ block is bijective and preserves the Hilbert-Schmidt inner product. It therefore defines a unitary change of basis in the corresponding complex operator space.

The decomposition of this complex operator space according to Hermiticity yields two real vector spaces: the subspace of Hermitian operators and the subspace of anti-Hermitian operators. Each has the same real dimension as the complex dimension of $\mathrm{End}(\mathcal{H})$, and the two subspaces are related by multiplication by $i$. Since the Hermitian adjoint of each real multipole tensor component is again a real multipole tensor component, see eq.~\eqref{eq:hermconj_realmultipoles}, one can form Hermitian and anti-Hermitian combinations of each adjoint-related pair. This construction preserves mutual orthogonality. For components with $l\neq l'$, however, these combinations have one half of the Hilbert--Schmidt norm squared of the original inter-shell operator. This change of norm is accounted for explicitly by the factor $2^{\delta_{l\neq l'}}$ in the reconstruction formula.

Due to the definite time-reversal symmetry of the real Hermitian multipole operators, a further decomposition by time-reversal symmetry does not double the number of operators, see Table \ref{tab:hermitian_time_selection_rule}. 

By virtue of matching dimensions, orthogonality of the spherical multipole tensors and all subsequent transformations being orthogonal, the time-reversal even and odd real Hermitian multipole operators $W^{\nu kpr,\mathrm{real}}_{l_1,l_2,t}$ form a complete orthogonal real basis of $\mathrm{End}(\mathcal{H})_\mathrm{Herm}$, and after allowing for complex coefficients, a basis of $\mathrm{End}(\mathcal{H})$.

\subsection{Time-reversal of operators}
\label{app:timereversed_operators}
In the following we determine the time-reversed operators of the spherical, spin, and multipole operators as well as the density matrix.
As described in the appendix \ref{app:timereversal}, in particular with claim \ref{claim:timereversal}, we determine the effect of the time-reversal operator by its action on the basis.\\

The typical convention for the orbital states is $\Theta|l,m\rangle = (-1)^{m} |l,-m\rangle$, 
as in Sakurai's Modern Quantum Mechanics \cite{sakurai}.
It reflects complex conjugation of the spherical harmonics in the Condon-Shortley phase convention.
The application of $\Theta$ on the Pauli matrices is given by $\Theta\vec{\sigma}\Theta^{-1} = -\vec{\sigma}$ . With the choice of spherical spin operator as in eq.~\eqref{eq:spherical_spin_operator}, it implies that the spin states transform as $\Theta|s,m_s\rangle = (-1)^{s-m_s} |s,-m_s\rangle$ \cite{berkeley_timereversal}. For the real components of the spin tensor operator we have $\Theta \sigma^{p, \mathrm{real}}_\alpha \Theta^{-1}=(-1)^{p\neq 0}\sigma^{p, \mathrm {real}}_\alpha$, and for the spherical components
\begin{equation}
    \Theta \sigma^p_y\Theta^{-1} = (-1)^{p-y}\sigma^p_{-y}
    \label{eq:timereversed_spin_tensor}
\end{equation}
To show this, one can use $\Theta = K \sigma_y$ on the spin operators in spherical tensor form as written in \ref{sec:basis}, as is described by \cite{dresselhaus},
or equivalently $\Theta = i \sigma_y K$ up to an overall phase, as described in \cite{sakurai}, where $K$ is the complex-conjugate operator.

\subsubsection{Time-reversal of the orbital tensor}
\label{app:timerev_orbitaltensor}

Using the Condon-Shortley convention
\begin{equation}
\Theta |lm\rangle = (-1)^m |l,-m\rangle ,
\end{equation}
and the antiunitarity of \(\Theta\), we obtain
\begin{align}
&\Theta C^k_q(l',l)\Theta^{-1}\\
&=
\sum_{m,m'}
n_{lkl'}^{-1}
(-1)^{l-m}
\begin{pmatrix}
l & k & l'\\
-m & q & m'
\end{pmatrix}
\\
& \hspace{3em}\times
(-1)^{m+m'}
|l,-m\rangle \langle l',-m'|
\nonumber\\
&=
\sum_{m,m'}
n_{lkl'}^{-1}
(-1)^{l+m}
\begin{pmatrix}
l & k & l'\\
m & q & -m'
\end{pmatrix}
\\
& \hspace{3em}\times
(-1)^{m+m'}
|lm\rangle \langle l'm'| .
\end{align}
Using the Wigner-3j symmetry
\begin{equation}
\begin{pmatrix}
l & k & l'\\
m & q & -m'
\end{pmatrix}
=
(-1)^{l+k+l'}
\begin{pmatrix}
l & k & l'\\
-m & -q & m'
\end{pmatrix},
\end{equation}
this becomes
\begin{align}
&\Theta  C^k_q(l',l)\Theta^{-1}\\
&=
(-1)^{l+k+l'}
\sum_{m,m'}
n_{lkl'}^{-1}
(-1)^{l+m}\\
&\quad \times
\begin{pmatrix}
l & k & l'\\
-m & -q & m'
\end{pmatrix}
(-1)^{m+m'}
|lm\rangle \langle l'm'| .
\end{align}
The Wigner-3j symbol imposes
\begin{equation}
-m-q+m'=0 \to m'+m=q+2m,
\end{equation}
and hence
\begin{equation}
(-1)^{m+m'} = (-1)^{-q}.
\end{equation}
Therefore
\begin{align}
&\Theta C^k_q(l',l)\Theta^{-1}\\
&=
(-1)^{l+k+l'}(-1)^{-q}\\
&\quad \times
\sum_{m,m'}
n_{lkl'}^{-1}
(-1)^{l-m}
\begin{pmatrix}
l & k & l'\\
-m & -q & m'
\end{pmatrix}
|lm\rangle \langle l'm'|
\nonumber\\
&=
(-1)^{l+l'+k-q}
 C^k_{-q}(l',l).
 \label{eq:timereversed_orbital_operator}
\end{align}

Time-reversal therefore maps the tensor component \(q\) to \(-q\), but it
does not exchange \(l\) and \(l'\). Time-reversal symmetry thus relates the $q$ and $-q$ components within the same $(l',l)$ block.

\subsubsection{Time-reversal of the spherical multipole operators}
\label{app:timerev_sphericalmultipoleoperators}

Using Eqs.~\eqref{eq:timereversed_spin_tensor} and
\eqref{eq:timereversed_orbital_operator}, we can determine the
time-reversal transformation of the coupled multipole tensor. For the multipole tensor components as defined in equation \eqref{eq:multipole_tensor}, we obtain
\begin{equation}
\begin{aligned}
&\Theta W^{kpr}_t(l',l)\Theta^{-1}\\
&=
\overline{n_{kpr}^{-1}}
(-1)^{k+p}
\sum_{q,y}
(-1)^{-q-y}
\begin{pmatrix}
k & r & p\\
-q & t & -y
\end{pmatrix}\\
&\qquad \times
(-1)^{l+l'+k-q}
(-1)^{p-y}
C^k_{-q}(l',l)\otimes \sigma^p_{-y}
\nonumber\\
&=(-1)^{l+l'}
\overline{n_{kpr}^{-1}}
\sum_{q,y}
\begin{pmatrix}
k & r & p\\
-q & t & -y
\end{pmatrix}
C^k_{-q}(l',l)\otimes \sigma^p_{-y}.
\end{aligned}
\end{equation}
Relabeling \(q\mapsto -q\) and \(y\mapsto -y\) gives
\begin{equation}
\begin{aligned}
&\Theta W^{kpr}_t(l',l)\Theta^{-1}\\
&=(-1)^{l+l'}
\overline{n_{kpr}^{-1}}
\sum_{q,y}
\begin{pmatrix}
k & r & p\\
q & t & y
\end{pmatrix}
C^k_q(l',l)\otimes \sigma^p_y .
\end{aligned}
\end{equation}
Using the Wigner-3j symmetry
\begin{equation}
\begin{pmatrix}
k & r & p\\
q & t & y
\end{pmatrix}
=
(-1)^{k+r+p}
\begin{pmatrix}
k & r & p\\
-q & -t & -y
\end{pmatrix},
\end{equation}
and $\overline{n_{kpr}^{-1}}=(-1)^{-k-p-r}n_{kpr}^{-1}$
we find
\begin{equation}
\begin{aligned}
&\Theta W^{kpr}_t(l',l)\Theta^{-1}\\
&=(-1)^{l+l'}
n_{kpr}^{-1}
\sum_{q,y}
\begin{pmatrix}
k & r & p\\
-q & -t & -y
\end{pmatrix}
C^k_q(l',l)\otimes \sigma^p_y .
\end{aligned}
\end{equation}
The Wigner-3j symbol in \(W^{kpr}_{-t}\) imposes
\begin{equation}
-q-t-y=0,
\end{equation}
and therefore
\begin{equation}
\begin{aligned}
&\Theta W^{kpr}_t(l',l)\Theta^{-1}=(-1)^{l+l'+t}
n_{kpr}^{-1}\\
&\hspace{1em}\times \sum_{q,y}(-1)^{-q-y}
\begin{pmatrix}
k & r & p\\
-q & -t & -y
\end{pmatrix}
C^k_q(l',l)\otimes \sigma^p_y .
\end{aligned}
\end{equation}
With $k$ and $p$ being integers, we can further introduce
\begin{equation}
1=(-1)^{k+p}(-1)^{k+p}
\end{equation}
and obtain

\begin{equation}
\Theta W^{kpr}_t(l',l)\Theta^{-1}
=
(-1)^{l+l'+k+p+t}
W^{kpr}_{-t}(l',l).
\label{eq:timereversed_multipole_tensor}
\end{equation}

\subsubsection{Time-reversal of the density matrix}
\label{app:timerev_densitymatrix}
We work in the product basis
\begin{equation}
|l,m\rangle |s, m_s\rangle
\equiv
|l,m\rangle \otimes |s, m_s\rangle ,
\end{equation}
and use the Condon-Shortley convention
\begin{equation}
\begin{aligned}
\Theta |l,m\rangle &= (-1)^m |l,-m\rangle,\\
\Theta |s, m_s\rangle &= (-1)^{s-m_s}|s,-m_s\rangle .
\end{aligned}
\end{equation}
Thus
\begin{equation}
\Theta \bigl(|l,m\rangle |s, m_s\rangle\bigr)
=
(-1)^{m+s-m_s}|l,-m\rangle |s,-m_s\rangle .
\end{equation}

Writing
\begin{equation}
\rho
=
\sum_{\substack{l m m_s\\ l' m' m_s'}}
\rho_{lm m_s,l'm'm_s'}
|l,m\rangle |s, m_s\rangle
\langle l',m'|\langle s, m_s'| ,
\end{equation}
antiunitarity of \(\Theta\) gives
\begin{align}
\tau(\rho)
&\equiv \Theta \rho \Theta^{-1}
\nonumber\\
&=
\sum_{\substack{l m m_s\\ l' m' m_s'}}
\overline{\rho_{lm m_s,l'm'm_s'}}
(-1)^{m+s-m_s}
(-1)^{m'+s-m_s'}
\nonumber\\
&\qquad\qquad\times
|l,-m\rangle |s,-m_s\rangle
\langle l',-m'|\langle s,-m_s'| .
\end{align}
Relabeling
\begin{equation}
\begin{aligned}
m\mapsto -m,\qquad
m'\mapsto -m',\\
m_s\mapsto -m_s,\qquad
m_s'\mapsto -m_s',
\end{aligned}
\end{equation}
and using Hermiticity of the density matrix,
\begin{equation}
\overline{
\rho_{l,-m,-m_s,l',-m',-m_s'}
}
=
\rho_{l',-m',-m_s',l,-m,-m_s},
\end{equation}
one obtains
\begin{equation}
\begin{aligned}
&\tau(\rho)_{lm m_s,l'm'm_s'}\\
&=
(-1)^{m+m'+2s+m_s+m_s'}
\rho_{l',-m',-m_s',l,-m,-m_s}.
\end{aligned}
\label{eq:tr_density_full_basis}
\end{equation}
For an electron, \(s=1/2\), so this becomes
\begin{equation}
\begin{aligned}
&\tau(\rho)_{lm m_s,l'm'm_s'}\\
&=
(-1)^{m+m'+1+m_s+m_s'}
\rho_{l',-m',-m_s',l,-m,-m_s}.
\end{aligned}
\end{equation}

We now rewrite the same transformation in spin channels. Let
\(\sigma^p_\alpha\) denote a real spin basis with
\begin{equation}
\sigma^0_0 = \mathbbm 1_s,
\qquad
\sigma^1_\alpha \in \{\sigma_x,\sigma_y,\sigma_z\}.
\end{equation}
Using
\begin{equation}
\Tr_s\!\left[
(\sigma^p_\alpha)^\dagger \sigma^{p'}_{\alpha'}
\right]
=
2\delta_{pp'}\delta_{\alpha\alpha'},
\end{equation}
we decompose
\begin{equation}
\rho
=
\sum_{p,\alpha}
\rho^p_\alpha \otimes \sigma^p_\alpha,
\qquad
\rho^p_\alpha
=
\frac{1}{2}
\Tr_s\!\left[
(\sigma^p_\alpha)^\dagger \rho
\right].
\end{equation}
Since the identity is time-reversal even and the Pauli matrices are
time-reversal odd,
\begin{equation}
\Theta_s \sigma^p_\alpha \Theta_s^{-1}
=
(-1)^p \sigma^p_\alpha ,
\end{equation}
we find
\begin{equation}
\Theta \rho \Theta^{-1}
=
\sum_{p,\alpha}
(-1)^p
\Theta_l \rho^p_\alpha \Theta_l^{-1}
\otimes
\sigma^p_\alpha .
\end{equation}

Now write
\begin{equation}
\rho^p_\alpha
=
\sum_{l m,l'm'}
\rho^p_{\alpha,lm,l'm'}
|l,m\rangle \langle l',m'| .
\end{equation}
Using Hermiticity of \(\rho^p_\alpha\),
\begin{equation}
\overline{
\rho^p_{\alpha,l,-m,l',-m'}
}
=
\rho^p_{\alpha,l',-m',l,-m},
\end{equation}
the orbital part transforms as
\begin{align}
\left(
\Theta_l \rho^p_\alpha \Theta_l^{-1}
\right)_{lm,l'm'}
&=
(-1)^{m+m'}
\rho^p_{\alpha,l',-m',l,-m}.
\end{align}
Therefore
\begin{equation}
\tau(\rho)^p_{\alpha,lm,l'm'}
=
(-1)^p
(-1)^{m+m'}
\rho^p_{\alpha,l',-m',l,-m}.
\label{eq:tr_density_real_spin_channel}
\end{equation}
Thus the spin-scalar channel \(p=0\) is even under time-reversal, while
the spin-vector channels \(p=1\) acquire an additional minus sign.

\subsection{Multipole decomposition of the density matrix}
\label{app:multipole_decomp}

\noindent
For an operator $O\in\mathrm{End}(\mathcal{H})$ and a basis $B_i$ orthonormal with respect to the inner product on $\mathrm{End}(\mathcal{H})$, the operator can be expressed in that basis as
\begin{equation}
    O = \sum_i \langle B_i, O \rangle B_i,
\end{equation}
where the Hilbert-Schmidt inner product is typically used:
\begin{equation}
    \langle B_i, O \rangle = \Tr\left[ B_i^\dagger O\right]
\end{equation}
We can thus decompose a density matrix $\rho\in\mathrm{End}(\mathcal{H})$ into its multipole moments:
\begin{equation}
    w^{\nu kpr,\mathrm{real}}_{l_1,l_2,t} = \Tr\left[W^{\nu kpr,\mathrm{real}}_{l_1,l_2,t} \rho\right],
    \label{eq:multipole_moments}
\end{equation}
where we used that the multipole operator $W^{kpr,\mathrm{real},\pm,\nu}_t(l',l)$ is Hermitian, and distinguish the multipole operator from its expectation value, i.e. its multipole moment, via upper- and lowercase letters.\\
Note, however, that the multipole operators as defined in \eqref{eq:timesymmetrized_multipoles} are orthogonal rather than orthonormal. A true representation within the multipole basis would be given by

\begin{equation}
\begin{aligned}
\rho = \sum_{l_1\leq l_2,k,p,r,t,\nu} &2^{\delta_{l\ne l'}}\frac{1}{2}|n_{kpr}|^2(2r+1)n_{l_1kl_2}^2(2k+1) \\ & \times w^{\nu kpr,\mathrm{real}}_{l_1,l_2,t}W^{\nu kpr,\mathrm{real}}_{l_1,l_2,t},
\end{aligned}
\end{equation}
where the factor of $2$ for $l\neq l'$ comes from the projection onto Hermitian components.

\subsection{Alternative time-symmetrization}
\label{app:alternative_timesymm}
The transformation of the spherical tensor components as described in section \ref{subsec:real} is a complex transformation, for which we will use the following shorthand
\begin{equation}
T_\alpha = \sum_{t} U_{\alpha, t} T_t
\end{equation}
As discussed in the appendix \ref{app:timereversal}, the time-reversal operator is an antiunitary endomorphism $\Theta$. With $U$ being a complex matrix, time-reversal $\tau(\cdot)$ and $U$ do not commute. 
\begin{equation*}
    \langle T_\alpha \rangle = \Tr[(T_{\alpha})^{\nu}\rho].
\end{equation*}
Using claim \ref{claim:hermitians_commute_in_trace}, we can swap the projection $(\cdot)^{\nu}$ from the Hermitian operator $T_\alpha$ to the Hermitian density matrix:
\begin{equation}
\Tr[(T_{\alpha})^{\nu}\rho]= \Tr[T_{\alpha}\rho^{\nu}],
\end{equation}
which also allows by linearity of the trace, to pull out the transformation:

\begin{equation}
\langle T_{\alpha} \rangle= \sum_t U_{\alpha,t}\Tr[ T_{t}\rho^{\nu}],
\label{eq:alternative_expcomp}
\end{equation}
with no Hermiticity restriction for the operator $T_{t}$.

This is how the time-reversal even and odd multipoles are computed in multipyles \cite{multipyles}.
The formula for the time-reversal even and odd parts of the density matrix is derived in the appendix \ref{app:timerev_densitymatrix}.

\subsection{Relation between different spherical multipole moments across different orbital blocks}
\label{app:additional_relation}

\noindent
As shown in the appendix \ref{app:orthogonality_sphericalmultipole} and \ref{app:timereversed_operators}, the spherical multipole operators form an orthogonal basis of the operator space and time-reversal relates different components of different $(l',l)$ blocks.
This relation can be leveraged to give insight into how different components of a density matrix of definite time-reversal symmetry relate to each other.\\
The expectation value of the spherical multipole operator for a given density matrix of definite time-reversal symmetry $\nu$, i.e. the spherical multipole moment, is computed as
\begin{equation}
    \Tr \left[ W^{kpr}_t(l',l) \rho^\nu \right],
\end{equation}
as in eq.~\eqref{eq:alternative_expcomp} for the computation of the real Hermitian multipole moments.

With $\tau(\rho^\nu)=(-1)^\nu\rho^\nu$, and claim \ref{claim:trace_tau}, we have
\begin{equation}
    \overline{\Tr \left[ W^{kpr}_t(l',l) \rho^\nu \right]} = (-1)^\nu \Tr \left[ \tau\left(W^{kpr}_t(l',l)\right) \rho^\nu \right].
\end{equation}
Using the time-reversal relation \eqref{eq:timereversed_multipole_tensor}, we get
\begin{equation}
    \overline{\Tr \left[ W^{kpr}_t(l',l) \rho^\nu \right]} = (-1)^{l+l'+k+p+\nu+t}  \Tr \left[ W^{kpr}_{-t}(l',l) \rho^\nu \right].
\end{equation}

Another way to compute the complex conjugate of an expectation value is through Hermitian conjugation $\overline{\Tr[O]}=\Tr[O^\dagger]$, which results in
\begin{equation}
    \overline{\Tr \left[ W^{kpr}_t(l',l) \rho^\nu \right]} = \Tr \left[ \left(W^{kpr}_t(l',l)\right)^\dagger \rho^\nu \right],
\end{equation}
where we cyclically permuted the arguments of the trace and used Hermiticity of $\rho^\nu$ [$\left(\Theta \rho\Theta^\dagger\right)^\dagger = (\Theta^\dagger)^\dagger \rho^\dagger \Theta^\dagger = \Theta \rho \Theta^\dagger\implies (\rho^\nu)^\dagger=\rho^\nu$].

Using the Hermitian adjoint of the spherical multipole tensor in eq. \eqref{eq:hermconj_multipoletensor}, we obtain 
\begin{equation}
    \overline{\Tr \left[ W^{kpr}_t(l',l) \rho^\nu \right]} = (-1)^{l-l'+t}\Tr \left[ W^{kpr}_{-t}(l,l') \rho^\nu \right],
\end{equation}
and together with complex conjugation via time-reversal, see claim \ref{claim:trace_theta}, we obtain a relation between the multipole moments of opposite $(l',l)$ blocks
\begin{equation}
    \Tr \left[ W^{kpr}_{t}(l',l) \rho^\nu \right] = (-1)^{k+p+\nu}\Tr \left[ W^{kpr}_{t}(l,l') \rho^\nu \right].
\end{equation}
For $l'=l$ this relation reduces to the statement about the $k+p$ time-reversal symmetry of the multipoles.
For $l'\neq l$, however, it gives a relation independent of global phases on the normalizations and coupling constants.

Using this relation on the expectation values of the real Hermitian multipole operators via eq. \eqref{eq:alternative_expcomp}, we obtain again the selection rule as laid out in Table \ref{tab:hermitian_time_selection_rule} which we deduced on an operator rather than expectation value level.

\section{Time-Reversal}
\label{app:timereversal}
This introductory discussion closely follows Shapiro’s lecture notes on topological aspects of condensed matter physics, based on lectures by G. M. Graf \cite{shapiro2016topological}.
By Wigner's theorem, any symmetry operation, such as time-reversal, is represented by a unitary or antiunitary operator.
For a time-reversal invariant system, we require that time-reversing a state that evolved in time is the same as evolving a time-reversed state backwards in time,
i.e. for a given propagation operator $U(t)$ and state $|\psi\rangle$, we require that $\Theta U(t)|\psi\rangle = U(-t)\Theta|\psi\rangle$, where $\Theta$ is the time-reversal operator.

For time independent Hamiltonians, we have $U(t) = e^{-iHt}$, and thus $\Theta i H = -i H \Theta$, which
 implies that either $\Theta$ is antiunitary and commutes with $H$, or that $\Theta$ is unitary and anticommutes with $H$.
The former case is the commonly chosen convention, as the latter case would require the spectrum of $H$ to be symmetric around zero, 
which is generally not the case.

Given a Hilbert space $\mathcal{H}$ with states $x$ and $y$, and the inner product denoted by $\langle \cdot, \cdot \rangle$, the definition of an antiunitary operator $A$ is
\begin{equation}
    \langle A x , A y \rangle = \overline{\langle x, y \rangle}
     \label{eq:antiunitarity},
\end{equation}
where the $\overline{\cdot}$ denotes complex conjugation.
We are writing the inner product as $\langle x, y \rangle$ instead of the more common $\langle x|y\rangle$, 
to make it clear which state the operator is acting on. 
The adjoint of an antiunitary operator is commonly defined by 
\begin{equation}
    \langle x, A^\dagger y \rangle = \overline{\langle A x, y \rangle},
     \label{eq:adjoint_antiunitary}
\end{equation}
which together with equation \eqref{eq:antiunitarity} implies the relation $A^\dagger A = \mathbb{I}$.
Conjugate-symmetry of the inner product further implies that $\langle A^\dagger x, y \rangle = \overline{\langle  x,  A y \rangle}$ [$\langle A^\dagger x, y \rangle = \overline{\langle y,A^\dagger x \rangle} = \langle A y, x \rangle = \overline{\langle  x,  A y \rangle}$], 
which together with equation \eqref{eq:antiunitarity} implies $AA^\dagger = \mathbb{I}$.

With 
\begin{equation}
A^\dagger A = AA^\dagger = \mathbb{I}
\end{equation}
we have 
\begin{equation}
A^\dagger = A^{-1}.
\end{equation}

We now want to know how the time-reversal operator acts on linear operators based on given action on the basis elements of the Hilbert spaces.
Assume a linear operator $O:\mathcal{H}\to\mathcal{H}'$, expanded in the basis of the Hilbert spaces:
\begin{equation}
    O = \sum_{i,j} O_{ij} |\varphi_i \rangle \langle \psi_j |
\end{equation}
We define the time-reversal of an operator via the time-reversal of the basis elements, for which the time-reversal is predefined. As shown below in claim \ref{claim:timereversal}, this corresponds to $\Theta O \Theta^\dagger$.

For an operator $O\in \mathrm{End}(\mathcal{H})$ we want to consider the decomposition into time even ($\nu=0$) and time odd ($\nu=1$) parts. We define the operator $P_{\nu}:\mathrm{End}(\mathcal{H})\to\mathrm{End}(\mathcal{H})$:
\begin{equation}
O^{\nu}=P_{\nu}(O) = \frac{1}{2}(O + (-1)^{\nu}\Theta O \Theta^\dagger),
\end{equation}
We further use the shorthand $\tau(O)$ for $\Theta O \Theta^\dagger$.
Note for $O\in\mathrm{Hom}(\mathcal{H'},\mathcal{H})$ we can then consider the extension to $\mathcal O\in\mathrm{End}(\mathcal{H'\oplus H})$, as done with the multipole moments in subsection \ref{sec:decomp_time}.

\begin{claim}
\label{claim:timereversal}
time-reversal of A can be treated as follows:
\begin{equation}
\Theta A \Theta^\dagger = \sum_{i,j} \overline{A_{ij}} \Theta |\varphi_i \rangle \langle \psi_j |\Theta^\dagger = \sum_{i,j} \overline{A_{ij}} |\Theta \varphi_i \rangle \langle \Theta \psi_j |
\end{equation}
i.e. $\Theta |\varphi\rangle \langle \psi | \Theta^\dagger = |\Theta \varphi\rangle \langle \Theta \psi |$.
\end{claim} 

\begin{proof} 
Let $\varphi,x \in \mathcal{H}'$ and $\psi,y\in\mathcal{H}$. Consider taking a matrix element $\langle x|  \Theta |\varphi\rangle \langle \psi | \Theta^\dagger  |y\rangle$. Resolving the right-most product first (Operators act from the left), using \eqref{eq:antiunitarity}:

\begin{equation}
      \langle \psi | \Theta^\dagger  |y\rangle =  \langle \psi , \Theta^\dagger y\rangle = \overline{\langle \Theta \psi ,  y\rangle} =\overline{\langle \Theta \psi |  y\rangle}
      \label{eq:exp_of_thetadagger}
\end{equation}
which now constitutes a complex coefficient of $|\varphi\rangle$. The time-reversal operator acting on $|\varphi\rangle$ acts on $|\varphi\rangle$ and on its coefficient \eqref{eq:exp_of_thetadagger} via complex conjugation:
\begin{equation}
    \langle x|  \Theta |\varphi\rangle\overline{\langle \Theta \psi |  y\rangle} = \langle x | \Theta \varphi\rangle\langle \Theta \psi |  y\rangle
\end{equation}
Since this hold for any $x,y$ the claim has been shown.
\end{proof}

\begin{claim}
\label{claim:trace_theta}
The trace of a time-reversed operator is the complex conjugate of the trace of the operator:
\begin{equation}
    \Tr[\Theta A \Theta ^\dagger]=\overline{\Tr[A]}
    \label{eq:trace_conj}
\end{equation}
\end{claim}
\begin{proof}
 Consider a single summand $\langle x | A |y\rangle$. For a given orthonormal basis of a Hilbert space, time-reversing each element yields again a basis (orthonormality is preserved by antiunitarity defined in equation \eqref{eq:antiunitarity}, and completeness: any $x=\Theta \Theta^\dagger x=\Theta \sum \langle \Theta^\dagger x, e_i\rangle e_i=\sum \overline{\langle \Theta^\dagger x, e_i\rangle}\Theta e_i=\sum \langle x, \Theta e_i\rangle \Theta e_i$ is spanned by $\{\Theta e_i\}_i$).
Since the trace is basis independent, we can take the trace of $\Theta A \Theta^\dagger$ in the time-reversed basis, which for each summand gives:
\begin{equation}
\begin{aligned}
    \langle \Theta x | \Theta A \Theta^\dagger | \Theta y \rangle &= \langle \Theta x , \Theta A \Theta^\dagger  \Theta y \rangle \\
    &= \langle \Theta x , \Theta Ay \rangle \\
    & = \overline{\langle x ,  Ay \rangle}= \overline{\langle x |  A |y \rangle}.
\end{aligned}
\end{equation}
\end{proof}

\begin{claim}
\label{claim:trace_tau}
Given two operators $A,B\in\mathrm{End}(\mathcal{H})$:
    \begin{equation}
    \Tr[\tau(A) B]=\overline{\Tr[A \tau(B)]}
\end{equation}
\end{claim}
\begin{proof}
Let $\Theta^2=\epsilon\,1$, with $\epsilon=\pm1$. 
\begin{equation}
\Theta \Theta^\dagger=1 \implies \Theta \Theta \Theta^\dagger \Theta^\dagger=1\implies \epsilon\Theta^\dagger \Theta^\dagger=1, 
\end{equation}
and we get $(\Theta^\dagger)^2=\epsilon\,1$. We thus get 
\begin{equation}
\begin{aligned}
    \overline{\Tr[A\Theta B \Theta^\dagger]}&=\Tr[\Theta A\Theta B \Theta^\dagger\Theta^\dagger]\\
    &=\epsilon\Tr[\Theta A\Theta B]\\
    &=\epsilon\Tr[\Theta A\Theta  \Theta \Theta^\dagger B]\\
    &=\Tr[\Theta A \Theta^\dagger B],
\end{aligned}
\end{equation}
where we used claim \ref{claim:trace_theta} in the first equality. 
\end{proof}

\begin{claim}
\label{claim:trace_tau_herm}
Given two Hermitian operators $A,B\in\mathrm{End}(\mathcal{H})$:
    \begin{equation}
    \Tr[\tau(A) B]=\Tr[A \tau(B)]
\end{equation}
\end{claim}
\begin{proof}
As $A$ and $B$ are Hermitian, so are $\tau(A)$ and $\tau(B)$:
\begin{equation}
\begin{aligned}
    \langle \Theta A \Theta^\dagger x, y\rangle &=\overline{\langle A \Theta^\dagger x, \Theta^\dagger y\rangle}\\
    &=\overline{\langle \Theta^\dagger x,A \Theta^\dagger y\rangle}\\
    &=\langle x, \Theta A\Theta^\dagger y\rangle
\end{aligned}
\label{eq:tauAishermitian}
\end{equation}
where in the first and last equality antiunitarity of $\Theta$ and $\Theta^\dagger$ were used, see equations \eqref{eq:antiunitarity} and \eqref{eq:adjoint_antiunitary}, and the Hermiticity of $A$ in the second equality.

Claim \ref{claim:trace_tau} gives:
\begin{equation}
\begin{aligned}
    \Tr[\tau(A) B] &=\overline{\Tr[A \tau(B)]}
\end{aligned}
\end{equation}
With $A$ and $\tau(B)$ Hermitian, we have 
\begin{equation}
\begin{aligned}
    \overline{\Tr[A \tau(B)]} &= \Tr\left[(A\tau(B))^\dagger \right]\\
    &=\Tr[\tau(B)A]
    =\Tr[A\tau(B)],
\end{aligned}
\end{equation}
where we have used cyclicity of the trace in the last step.
\end{proof}

\begin{claim}
\label{claim:hermitians_commute_in_trace}
For Hermitian operators $A$ and $B$:
\begin{equation}
\Tr[A^\nu B] = \Tr[A B^{\nu}]
\end{equation}
\end{claim}

\begin{proof}

Using the definitions, linearity of the trace and claim \ref{claim:trace_tau_herm} we get:
\begin{equation}
\begin{aligned}
    \Tr[A^\nu B] &= \Tr[\frac{1}{2}(A+(-1)^{\nu}\tau(A))B]\\
    &= \frac{1}{2}(\Tr[AB]+(-1)^{\nu}\Tr[\tau(A)B])\\
    &= \frac{1}{2}(\Tr[AB]+(-1)^{\nu}\Tr[A\tau(B)])\\
    &= \Tr[A\frac{1}{2}(B+(-1)^{\nu}\tau(B))] = \Tr[A B^{\nu}]
\end{aligned}
\end{equation}
\end{proof}

\begin{claim}
$P_{\nu}$ is a projection, i.e. $P_{\nu}P_{\nu}=P_{\nu}$
\end{claim}
\begin{proof}
First note, that $\Theta A^\nu\Theta^\dagger=(-1)^{\nu}A^\nu$:
\begin{equation}
\begin{aligned}
\Theta A^\nu\Theta^\dagger &= \frac{1}{2}(\Theta A\Theta^\dagger + (-1)^{\nu}\Theta\Theta A\Theta^\dagger\Theta^\dagger)\\
&=(-1)^{\nu}\frac{1}{2}((-1)^{\nu}\Theta A\Theta^\dagger +  A)\\
&=(-1)^{\nu}A^\nu,
\end{aligned}
\end{equation}
where we pulled out the factor $(-1)^{\nu}$ and used $\Theta\Theta=\Theta^\dagger\Theta^\dagger=\epsilon\,1= \pm1$, which was shown in the proof for claim \ref{claim:trace_tau}.
Hence we get
\begin{equation}
\begin{aligned}
P_{\nu}P_{\nu}A=P_{\nu}A^\nu &=\frac{1}{2}(A^\nu+(-1)^{\nu}\Theta A^\nu\Theta^\dagger)\\
&=\frac{1}{2}(A^\nu+((-1)^{\nu})^2A^\nu) = A^\nu = P_{\nu}A,
\end{aligned}
\end{equation}
i.e. $P_{\nu}$ is a projection on the real operator space of Hermitian operators.
\end{proof}

\vspace{3cm}
\bibliographystyle{apsrev4-2}
\bibliography{references}

@article{Bultmark2009,
  author  = {Fredrik Bultmark and Francesco Cricchio and Oscar Gr{\aa}n{\"a}s and Lars Nordstr{\"o}m},
  title   = {Multipole decomposition of {LDA} + {U} energy and its application to actinide compounds},
  journal = {Phys. Rev. B},
  volume  = {80},
  number  = {3},
  pages   = {035121},
  year    = {2009},
  doi     = {10.1103/PhysRevB.80.035121}
}

@misc{multipyles,
author = {Merkel, M. E.},
note = {{\tt multipyles} v1.1.0 (Zenodo, 2023)
doi:10.5281/zenodo.8199391}
}

@book{sakurai,
  author    = {J. J. Sakurai and Jim Napolitano},
  title     = {Modern Quantum Mechanics},
  edition   = {2},
  publisher = {Cambridge University Press},
  address   = {Cambridge},
  year      = {2017},
  isbn      = {9781108422413}
}

@book{dresselhaus,
  author    = {Dresselhaus, Mildred S. and Dresselhaus, Gene and Jorio, Ado},
  title     = {Group Theory: Application to the Physics of Condensed Matter},
  publisher = {Springer},
  address   = {Berlin Heidelberg},
  year      = {2010},
  isbn      = {978-3-642-06945-1}
}

@book{Edmonds1957,
  author    = {Edmonds, A. R.},
  title     = {Angular Momentum in Quantum Mechanics},
  series    = {Investigations in Physics},
  number    = {4},
  publisher = {Princeton University Press},
  address   = {Princeton, NJ},
  year      = {1957}
}

@misc{berkeley_timereversal,
  author = {Littlejohn, R. G.},
  title = {Time Reversal},
  year  = {1996--1997},
  url   = {https://bohr.physics.berkeley.edu/classes/221/notes/timerev.pdf},
  note  = {Physics 221 Lecture Notes, University of California, Berkeley}
}

@article{Blochl1994PAW,
  author  = {Bl{\"o}chl, P. E.},
  title   = {Projector augmented-wave method},
  journal = {Phys. Rev. B},
  volume  = {50},
  number  = {24},
  pages   = {17953--17979},
  year    = {1994},
  doi     = {10.1103/PhysRevB.50.17953}
}

@article{KresseJoubert1999PAW,
  author  = {Kresse, G. and Joubert, D.},
  title   = {From ultrasoft pseudopotentials to the projector augmented-wave method},
  journal = {Phys. Rev. B},
  volume  = {59},
  number  = {3},
  pages   = {1758--1775},
  year    = {1999},
  doi     = {10.1103/PhysRevB.59.1758}
}

@article{Duros_2025,
  author  = {Duros, O and Juhin, A and Elnaggar, H and Chiuzbăian, G S and Brouder, C},
  title   = {General expressions for {Stevens} and {Racah} operator equivalents},
  journal = {J. Phys. A: Math. Theor.},
  volume  = {58},
  number  = {2},
  pages   = {025207},
  year    = {2024},
  month   = dec,
  doi     = {10.1088/1751-8121/ad96fc}
}

@article{Racah1942,
  author  = {Racah, Giulio},
  title   = {Theory of Complex Spectra. II},
  journal = {Phys. Rev.},
  volume  = {62},
  number  = {9-10},
  pages   = {438--462},
  year    = {1942},
  month   = nov,
  doi     = {10.1103/PhysRev.62.438}
}

@article{vanderLaan1995,
  author  = {van der Laan, Gerrit and Thole, B. T.},
  title   = {Core hole polarization in resonant photoemission},
  journal = {J. Phys.: Condens. Matter},
  volume  = {7},
  number  = {50},
  pages   = {9947},
  year    = {1995},
  month   = dec,
  doi     = {10.1088/0953-8984/7/50/028}
}

@article{BLANCO1997,
  author  = {Miguel A. Blanco and M. Flórez and M. Bermejo},
  title   = {Evaluation of the rotation matrices in the basis of real spherical harmonics},
  journal = {J. Mol. Struct.: THEOCHEM},
  volume  = {419},
  number  = {1},
  pages   = {19--27},
  year    = {1997},
  doi     = {10.1016/S0166-1280(97)00185-1}
}

@book{varshalovich1988,
  author    = {Varshalovich, D. A. and Moskalev, A. N. and Khersonskii, V. K.},
  title     = {Quantum Theory of Angular Momentum},
  publisher = {World Scientific},
  address   = {Singapore},
  year      = {1988},
  isbn      = {9789971501075}
}

@article{Spaldin2008,
  author  = {Spaldin, Nicola A and Fiebig, Manfred and Mostovoy, Maxim},
  title   = {The toroidal moment in condensed-matter physics and its relation to the magnetoelectric
effect*},
  journal = {J. Phys.: Condens. Matter},
  volume  = {20},
  number  = {43},
  pages   = {434203},
  year    = {2008},
  month   = oct,
  doi     = {10.1088/0953-8984/20/43/434203}
}

@article{Hayami2018PRB,
  author  = {Hayami, Satoru and Yatsushiro, Megumi and Yanagi, Yuki and Kusunose, Hiroaki},
  title   = {Classification of atomic-scale multipoles under crystallographic point groups and application to linear response tensors},
  journal = {Phys. Rev. B},
  volume  = {98},
  number  = {16},
  pages   = {165110},
  year    = {2018},
  month   = oct,
  doi     = {10.1103/physrevb.98.165110}
}

@article{Hayami2018,
  author  = {Hayami, Satoru and Kusunose, Hiroaki},
  title   = {Microscopic Description of Electric and Magnetic Toroidal Multipoles in Hybrid Orbitals},
  journal = {J. Phys. Soc. Jpn.},
  volume  = {87},
  number  = {3},
  pages   = {033709},
  year    = {2018},
  month   = mar,
  doi     = {10.7566/jpsj.87.033709}
}

@article{Kusunose2020,
  author  = {Kusunose, Hiroaki and Oiwa, Rikuto and Hayami, Satoru},
  title   = {Complete Multipole Basis Set for Single-Centered Electron Systems},
  journal = {J. Phys. Soc. Jpn.},
  volume  = {89},
  number  = {10},
  pages   = {104704},
  year    = {2020},
  doi     = {10.7566/JPSJ.89.104704}
}

@article{Yatsushiro2021,
  author  = {Yatsushiro, Megumi and Kusunose, Hiroaki and Hayami, Satoru},
  title   = {Multipole classification in 122 magnetic point groups for unified understanding of multiferroic responses and transport phenomena},
  journal = {Phys. Rev. B},
  volume  = {104},
  number  = {5},
  pages   = {054412},
  year    = {2021},
  month   = aug,
  doi     = {10.1103/physrevb.104.054412}
}

@article{Ederer2007,
  author  = {Ederer, Claude and Spaldin, Nicola A.},
  title   = {Towards a microscopic theory of toroidal moments in bulk periodic crystals},
  journal = {Phys. Rev. B},
  volume  = {76},
  number  = {21},
  pages   = {214404},
  year    = {2007},
  month   = dec,
  doi     = {10.1103/physrevb.76.214404}
}

@article{Inda2024,
  author  = {Inda, A. and Oiwa, R. and Hayami, S. and Yamamoto, H. M. and Kusunose, H.},
  title   = {Quantification of chirality based on electric toroidal monopole},
  journal = {J. Chem. Phys.},
  volume  = {160},
  number  = {18},
  pages   = {184117},
  year    = {2024},
  month   = may,
  doi     = {10.1063/5.0204254}
}

@article{Oiwa2025PRR,
  author  = {Oiwa, Rikuto and Kusunose, Hiroaki},
  title   = {Predominant electronic order parameter for structural chirality: Role of spinless electronic toroidal multipoles in {Te} and {Se}},
  journal = {Phys. Rev. Res.},
  volume  = {7},
  number  = {3},
  pages   = {033250},
  year    = {2025},
  month   = sep,
  doi     = {10.1103/1zq8-pqh8}
}

@article{Hayami2024PSJ,
  author  = {Hayami, Satoru and Kusunose, Hiroaki},
  title   = {Unified Description of Electronic Orderings and Cross Correlations by Complete Multipole Representation},
  journal = {J. Phys. Soc. Jpn.},
  volume  = {93},
  number  = {7},
  pages   = {072001},
  year    = {2024},
  doi     = {10.7566/JPSJ.93.072001}
}

@article{Furukawa2017,
  author  = {Furukawa, Tetsuya and Shimokawa, Yuri and Kobayashi, Kaya and Itou, Tetsuaki},
  title   = {Observation of current-induced bulk magnetization in elemental tellurium},
  journal = {Nat. Commun.},
  volume  = {8},
  number  = {1},
  pages   = {954},
  year    = {2017},
  doi     = {10.1038/s41467-017-01093-3}
}

@article{Calavalle2022,
  author  = {Calavalle, Francesco and Suárez-Rodríguez, Manuel and Martín-García, Beatriz and Johansson, Annika and Vaz, Diogo C. and Yang, Haozhe and Maznichenko, Igor V. and Ostanin, Sergey and Mateo-Alonso, Aurelio and Chuvilin, Andrey and Mertig, Ingrid and Gobbi, Marco and Casanova, Fèlix and Hueso, Luis E.},
  title   = {Gate-tuneable and chirality-dependent charge-to-spin conversion in tellurium nanowires},
  journal = {Nat. Mater.},
  volume  = {21},
  number  = {5},
  pages   = {526--532},
  year    = {2022},
  month   = mar,
  doi     = {10.1038/s41563-022-01211-7}
}

@article{Roy2022,
  author  = {Roy, Arunesh and Cerasoli, Frank T. and Jayaraj, Anooja and Tenzin, Karma and Nardelli, Marco Buongiorno and Sławińska, Jagoda},
  title   = {Long-range current-induced spin accumulation in chiral crystals},
  journal = {npj Comput. Mater.},
  volume  = {8},
  number  = {1},
  pages   = {243},
  year    = {2022},
  month   = nov,
  doi     = {10.1038/s41524-022-00931-3}
}

@book{BrinkSatchler1968,
  author    = {Brink, D. M. and Satchler, G. R.},
  title     = {Angular Momentum},
  publisher = {Clarendon Press},
  address   = {Oxford},
  year      = {1968},
  pages     = {160}
}

@book{Wigner1959,
  author     = {Wigner, Eugene P.},
  title      = {Group Theory and Its Application to the Quantum Mechanics of Atomic Spectra},
  publisher  = {Academic Press},
  address    = {New York},
  year       = {1959},
  pages      = {ix+372},
  translator = {Griffin, J. J.}
}

@article{BhowalSpaldin2021,
  author  = {Bhowal, Sayantika and Spaldin, Nicola A.},
  title   = {Revealing hidden magnetoelectric multipoles using {Compton} scattering},
  journal = {Phys. Rev. Res.},
  volume  = {3},
  number  = {3},
  pages   = {033185},
  year    = {2021},
  month   = aug,
  doi     = {10.1103/physrevresearch.3.033185}
}

@article{Sakano2020,
  author  = {Sakano, M. and Hirayama, M. and Takahashi, T. and Akebi, S. and Nakayama, M. and Kuroda, K. and Taguchi, K. and Yoshikawa, T. and Miyamoto, K. and Okuda, T. and Ono, K. and Kumigashira, H. and Ideue, T. and Iwasa, Y. and Mitsuishi, N. and Ishizaka, K. and Shin, S. and Miyake, T. and Murakami, S. and Sasagawa, T. and Kondo, Takeshi},
  title   = {Radial Spin Texture in Elemental Tellurium with Chiral Crystal Structure},
  journal = {Phys. Rev. Lett.},
  volume  = {124},
  number  = {13},
  pages   = {136404},
  year    = {2020},
  month   = mar,
  doi     = {10.1103/physrevlett.124.136404}
}

@article{Gatti2020,
  author  = {Gatti, G. and Gos\'albez-Mart\'{\i}nez, D. and Tsirkin, S. S. and Fanciulli, M. and Puppin, M. and Polishchuk, S. and Moser, S. and Testa, L. and Martino, E. and Roth, S. and Bugnon, Ph. and Moreschini, L. and Bostwick, A. and Jozwiak, C. and Rotenberg, E. and Di Santo, G. and Petaccia, L. and Vobornik, I. and Fujii, J. and Wong, J. and Jariwala, D. and Atwater, H. A. and R\o{}nnow, H. M. and Chergui, M. and Yazyev, O. V. and Grioni, M. and Crepaldi, A.},
  title   = {Radial Spin Texture of the {Weyl} Fermions in Chiral Tellurium},
  journal = {Phys. Rev. Lett.},
  volume  = {125},
  number  = {21},
  pages   = {216402},
  year    = {2020},
  month   = nov,
  doi     = {10.1103/PhysRevLett.125.216402}
}

@article{Cricchio2009,
  author  = {Cricchio, Francesco and Bultmark, Fredrik and Gr\aa{}n\"as, Oscar and Nordstr\"om, Lars},
  title   = {Itinerant Magnetic Multipole Moments of Rank Five as the Hidden Order in {${\mathrm{URu}}_{2}{\mathrm{Si}}_{2}$}},
  journal = {Phys. Rev. Lett.},
  volume  = {103},
  number  = {10},
  pages   = {107202},
  year    = {2009},
  month   = sep,
  doi     = {10.1103/PhysRevLett.103.107202}
}

@article{Spaldin2013,
  author  = {Spaldin, Nicola A. and Fechner, Michael and Bousquet, Eric and Balatsky, Alexander and Nordstr\"om, Lars},
  title   = {Monopole-based formalism for the diagonal magnetoelectric response},
  journal = {Phys. Rev. B},
  volume  = {88},
  number  = {9},
  pages   = {094429},
  year    = {2013},
  month   = sep,
  doi     = {10.1103/PhysRevB.88.094429}
}

@article{TholeFechnerSpaldin2016,
  author  = {Fechner, M. and Spaldin, N. A.},
  title   = {Effects of intense optical phonon pumping on the structure and electronic properties of yttrium barium copper oxide},
  journal = {Phys. Rev. B},
  volume  = {94},
  number  = {13},
  pages   = {134307},
  year    = {2016},
  month   = oct,
  doi     = {10.1103/PhysRevB.94.134307}
}

@article{Thole2018,
  author  = {Th{\"o}le, Florian and Spaldin, Nicola A.},
  title   = {Magnetoelectric multipoles in metals},
  journal = {Philos. Trans. R. Soc. A},
  volume  = {376},
  pages   = {20170450},
  year    = {2018},
  doi     = {10.1098/rsta.2017.0450}
}

@article{UrruSpaldin2022,
  author  = {Urru, Andrea and Spaldin, Nicola A.},
  title   = {Magnetic octupole tensor decomposition and second-order magnetoelectric effect},
  journal = {Ann. Phys.},
  volume  = {447},
  pages   = {168964},
  year    = {2022},
  month   = dec,
  doi     = {10.1016/j.aop.2022.168964}
}

@article{Urru2023,
  author  = {Urru, Andrea and Soh, Jian-Rui and Qureshi, Navid and Stunault, Anne and Roessli, Bertrand and R{\o}nnow, Henrik M. and Spaldin, Nicola A.},
  title   = {Neutron scattering from local magnetoelectric multipoles: A combined theoretical, computational, and experimental perspective},
  journal = {Phys. Rev. Res.},
  volume  = {5},
  number  = {3},
  pages   = {033147},
  year    = {2023},
  month   = sep,
  doi     = {10.1103/PhysRevResearch.5.033147}
}

@article{VerbeekUrruSpaldin2023,
  author  = {Verbeek, Xanthe H. and Urru, Andrea and Spaldin, Nicola A.},
  title   = {Hidden orders and (anti-)magnetoelectric effects in {Cr$_2$O$_3$} and {$\alpha$-Fe$_2$O$_3$}},
  journal = {Phys. Rev. Res.},
  volume  = {5},
  number  = {4},
  pages   = {l042018},
  year    = {2023},
  month   = nov,
  doi     = {10.1103/physrevresearch.5.l042018}
}

@article{Spaldin2026,
  author  = {Nicola A. Spaldin},
  title   = {Toward a modern theory of chiralization},
  journal = {Newton},
  pages   = {100581},
  year    = {2026},
  doi     = {10.1016/j.newton.2026.100581}
}

@article{Tsirkin2018,
  author  = {Tsirkin, Stepan S. and Puente, Pablo Aguado and Souza, Ivo},
  title   = {Gyrotropic effects in trigonal tellurium studied from first principles},
  journal = {Phys. Rev. B},
  volume  = {97},
  number  = {3},
  pages   = {035158},
  year    = {2018},
  month   = jan,
  doi     = {10.1103/physrevb.97.035158}
}

@misc{shapiro2016topological,
  author  = {Shapiro, Jacob},
  title   = {Notes on Topological Aspects of Condensed Matter Physics},
  year    = {2016},
  url     = {https://web.math.princeton.edu/~js129/PDFs/Top\_SSP\_Lecture\_Notes.pdf},
  urldate = {2026-06-09},
  note    = {Based on notes by Prof. G. M. Graf}
}

@article{Kresse:1996,
  author  = {Kresse, G. and Furthm\"uller, J.},
  title   = {Efficient iterative schemes for ab initio total-energy calculations using a plane-wave basis set},
  journal = {Phys. Rev. B},
  volume  = {54},
  number  = {16},
  pages   = {11169--11186},
  year    = {1996},
  month   = oct,
  doi     = {10.1103/PhysRevB.54.11169}
}

@article{PBE:1996,
  author  = {Perdew, John P. and Burke, Kieron and Ernzerhof, Matthias},
  title   = {Generalized Gradient Approximation Made Simple},
  journal = {Phys. Rev. Lett.},
  volume  = {77},
  number  = {18},
  pages   = {3865--3868},
  year    = {1996},
  month   = oct,
  doi     = {10.1103/PhysRevLett.77.3865}
}

@article{IVDW21:2013,
  author  = {Bu{\v{c}}ko, Tom{\'a}{\v{s}}
and Leb{\`e}gue, S{\'e}bastien
and Hafner, J{\"u}rgen
and {\'A}ngy{\'a}n, J{\'a}nos G.},
  title   = {Improved Density Dependent Correction for the Description of {London} Dispersion Forces},
  journal = {J. Chem. Theory Comput.},
  volume  = {9},
  number  = {10},
  pages   = {4293--4299},
  year    = {2013},
  month   = oct,
  doi     = {10.1021/ct400694h}
}

@article{IVDW21:2014,
  author  = {Bučko, Tomáš and Lebègue, Sébastien and Ángyán, János G. and Hafner, Jürgen},
  title   = {Extending the applicability of the {Tkatchenko--Scheffler} dispersion correction via iterative {Hirshfeld} partitioning},
  journal = {J. Chem. Phys.},
  volume  = {141},
  number  = {3},
  pages   = {034114},
  year    = {2014},
  month   = {07},
  doi     = {10.1063/1.4890003}
}

@article{VESTA_Momma:2008,
  author  = {Momma, Koichi and Izumi, Fujio},
  title   = {{VESTA}: a three-dimensional visualization system for electronic and structural analysis},
  journal = {J. Appl. Crystallogr.},
  volume  = {41},
  number  = {3},
  pages   = {653--658},
  year    = {2008},
  doi     = {10.1107/S0021889808012016}
}

@article{Oiwa_PRL:2022,
  author  = {Oiwa, Rikuto and Kusunose, Hiroaki},
  title   = {Rotation, Electric-Field Responses, and Absolute Enantioselection in Chiral Crystals},
  journal = {Phys. Rev. Lett.},
  volume  = {129},
  number  = {11},
  pages   = {116401},
  year    = {2022},
  month   = sep,
  doi     = {10.1103/PhysRevLett.129.116401}
}

@article{Kusunose_chirality_APL_2024,
  author  = {Kusunose, Hiroaki and Kishine, Jun-ichiro and Yamamoto, Hiroshi M.},
  title   = {Emergence of chirality from electron spins, physical fields, and material-field composites},
  journal = {Appl. Phys. Lett.},
  volume  = {124},
  number  = {26},
  pages   = {260501},
  year    = {2024},
  month   = {06},
  doi     = {10.1063/5.0214919}
}

@article{Choi1999,
  author  = {Choi, Cheol Ho and Ivanic, Joseph and Gordon, Mark S. and Ruedenberg, Klaus},
  title   = {Rapid and stable determination of rotation matrices between spherical harmonics by direct recursion},
  journal = {J. Chem. Phys.},
  volume  = {111},
  number  = {19},
  pages   = {8825--8831},
  year    = {1999},
  month   = {11},
  doi     = {10.1063/1.480229}
}

@article{Winkler,
  author  = {Winkler, R. and Z\"ulicke, U.},
  title   = {Theory of electric, magnetic, and toroidal polarizations in crystalline solids with applications to hexagonal lonsdaleite and cubic diamond},
  journal = {Phys. Rev. B},
  volume  = {107},
  number  = {15},
  pages   = {155201},
  year    = {2023},
  month   = apr,
  doi     = {10.1103/PhysRevB.107.155201}
}

@Article{Keller_Te:1977,
author={Keller, R.
and Holzapfel, W. B.
and Schulz, Heinz},
title={Effect of pressure on the atom positions in Se and Te},
journal={Phys. Rev. B},
year={1977},
month={Nov},
day={15},
publisher={American Physical Society},
volume={16},
number={10},
pages={4404-4412},
doi={10.1103/PhysRevB.16.4404},
}

\end{document}